\documentclass[11pt,letterpaper]{article}
\usepackage{amsmath,amssymb,amsfonts,amsthm}
\usepackage{enumitem}
\usepackage{bbm}
\usepackage{xcolor}
\usepackage[ruled,vlined,linesnumbered]{algorithm2e}
\usepackage[margin=1in]{geometry}
\usepackage[colorlinks, citecolor = blue, linkcolor = magenta]{hyperref}
\usepackage{cleveref}
\usepackage{float}
\usepackage{caption}
\usepackage{subcaption}
\usepackage{tikz}
\usetikzlibrary{shapes,arrows,chains}
\usetikzlibrary[calc]
\tikzstyle{plan}=[draw, rounded corners,align=center]
\tikzstyle{line} = [draw, -latex']

\numberwithin{equation}{section}

\DeclareMathOperator*{\E}{\mathbb{E}}
\DeclareMathOperator*{\Var}{\mathrm{Var}}
\DeclareMathOperator*{\PP}{\mathbb{P}}

\DeclareMathOperator*{\argmin}{arg\,min}

\renewcommand{\d}{\,\mathrm{d}}
\newcommand{\fake}{\mathsf{f}}
\newcommand{\true}{\mathsf{t}}

\newtheorem{theorem}{Theorem}[section]
\newtheorem{lemma}[theorem]{Lemma}
\theoremstyle{definition}
\newtheorem{remark}{Remark}
\newtheorem{proposition}[theorem]{Proposition}
\newtheorem{corollary}[theorem]{Corollary}
\newtheorem{definition}[theorem]{Definition}
\newtheorem{claim}{Claim}
\newtheorem{example}[theorem]{Example}
\allowdisplaybreaks

\title{Efficiently matching random inhomogeneous graphs\\ via degree profiles}

\author{Jian Ding\thanks{Partially supported by NSFC Key Program Project No. 12231002.}\\Peking University \and Yumou Fei\\Peking University \and Yuanzheng Wang\\Peking University}

\begin{document}

\maketitle

\begin{abstract}
	In this paper, we study the problem of recovering the latent vertex correspondence between two correlated random graphs with vastly inhomogeneous and unknown edge probabilities between different pairs of vertices. Inspired by and extending the matching algorithm via degree profiles by Ding, Ma, Wu and Xu (2021), we obtain an efficient matching algorithm as long as the minimal average degree is at least $\Omega(\log^{2} n)$ and the minimal correlation is at least $1 - O(\log^{-2} n)$. 
\end{abstract}
\section{Introduction}

In this paper we study the problem of recovering the latent matching between two correlated random graphs. A key novelty in our work, as illustrated in the following definition, is that we hugely relax the model assumptions in the correlated Erd\H{o}s--R\'enyi graph model.

\begin{definition}\label{def:graphensemble}
A pair of random graphs $(G,H)$ on $n$ vertices is said to be of an $(\alpha,d,\delta)$-edge-correlated distribution if the following hold:
\begin{itemize}[itemsep=0pt]
    \item (Matrix-encoded random graphs). Consider $G$ and $H$ as random symmetric $n\times n$ matrices. The random vectors $(G_{ij},H_{ij})$, for $1\leq i\leq j\leq n$, are mutually independent and take values in $\{0,1\}^{2}$. 
    \item (Inhomogeneous edge probabilities). There are fixed parameters $p_{ij}\in [0,1-\alpha]$ for $i,j\in [n]$, with $p_{ij}=p_{ji}$ and $p_{ii}=0$, such that
    $$\PP[G_{ij}=1]=\PP[H_{ij}=1]=p_{ij},\text{ for }1\leq i\leq j\leq n.$$
    \item (Edge correlations). There are fixed parameters $\delta_{ij}\in [0,\delta]$ for $i,j\in [n]$, with $\delta_{ij}=\delta_{ji}$ and $\delta_{ii}=0$, such that $$\PP[G_{ij}=1,\,H_{ij}=0]=\PP[G_{ij}=0,\,H_{ij}=1]=p_{ij}\delta_{ij}, \text{ for } 1\leq i \leq j\leq n.$$
    \item (Minimal average degree). For each $i\in [n]$, if we define $d_{i}=\sum_{j=1}^{n}p_{ij}$, then $d_{i} \geq d$.
\end{itemize}
\end{definition}

In \Cref{def:graphensemble}, the only minor assumption on the edge probabilities $p_{ij}$, $i,j\in [n]$ (here $[n]=\{1,2,\dots,n\}$) is that they are uniformly bounded away from $1$, quantified by the parameter $\alpha$. We pose this assumption to facilitate the analysis of our variance-stabilizing transformation; see \Cref{subsec:intro_comparison} for more detailed explanations. The edge correlations are $(1-\delta_{ij}-p)/(1-p)$ for $i,j\in [n]$, and thus all the edge correlations are at least $1-\delta$. The parameter $d_{i}$ is the expected degree of the $i$-th vertex in $G$ and $H$, and we assume that all the expected degrees are uniformly bounded from below by the minimal average degree $d$, so as to ``ensure sufficient data''.

In what follows, we showcase the generality of \Cref{def:graphensemble} by giving examples of specific models that fall into the setting, including the correlated Erd\H{o}s-R\'{e}nyi model, the correlated stochastic block model, and the correlated Chung-Lu model.

\begin{example}[Correlated Erd\H{o}s-R\'{e}nyi model $\mathbb{G}(n,q;s)$]\label{eg:correlated er}
Given a positive integer $n$ and $q,s\in [0,1]$, the correlated Erd\H{o}s-R\'{e}nyi model $\mathbb{G}(n,q;s)$ is defined as follows. First sample a parent Erd\H{o}s-R\'{e}nyi graph $A\sim \mathbb{G}(n,q/s)$. Conditioned on $A$, we then sample $G$ and $H$  as independent subgraphs of $A$ obtained by keeping each edge in $A$ with probability $s$ independently. Thus, $(G,H)$ is a pair of correlated Erd\H{o}s-R\'{e}nyi random graphs, each of which has marginal distribution $\mathbb{G}(n,q)$, and the correlations between corresponding edges are $(s-q)/(1-q)$. Then the law of $(G, H)$, denoted as $\mathbb G(n, q; s)$, is a special case of \Cref{def:graphensemble} by taking $p_{ij}=q, \delta_{ij}=1-s$ for $i,j\in [n]$ and taking $d_{i}=(n-1)q$ for $i\in [n]$.
\end{example}

\begin{example}[Correlated stochastic block models]\label{eg:correlated sbm}
Given a positive integer $r\leq n$, we partition the vertex set $[n]$ into disjoint subsets $C_{1},C_{2},\dots,C_{r}$, and also introduce a symmetric $r\times r$ matrix $P$ of edge probabilities. Then under the correlated stochastic block model, the pair of graphs $(G,H)$ is obtained from the parent graph $A$ via subsampling in the same way as in \Cref{eg:correlated er}, whereas $A$ is a random graph consisting of independent edges and edge probability between $u, v$ is $P_{ij}/s$ for different vertices $u,v$ with $u\in C_i$, $v\in C_j$ and $i, j\in [r]$. This model, which introduces inhomogenity to the correlated Erd\H{o}s-R\'{e}nyi model, is again a special case of \Cref{def:graphensemble} by choosing $p_{uv}=P_{ij}$ for $u\in C_{i}$ and $v\in C_{j}$.  
\end{example}

\begin{example}[Correlated Chung-Lu model]\label{eg:correlated cl}
Following the Chung-Lu model introduced in \cite{CL02}, we sample our random graph $A$ as follows: given an expected degree sequence $(w_{1},w_{2},\dots,w_{n})\in\mathbb{R}_{+}^{n}$ that satisfy $\max_{i}w_{i}^{2}<\sum_{i}w_{i}$, for $i, j\in [n]$ we connect an edge between vertices $i$ and $j$ independently with probability  $w_{ij}=w_{i}w_{j}/(\sum_{i}w_{i})$. Given the parent graph $A$, we again sample a pair of conditionally independent subgraphs $(G,H)$ via subsampling in the same manner as in \Cref{eg:correlated er}. Then the correlated Chung-Lu model, i.e., the law of $
(G, H)$, is also a special case of \Cref{def:graphensemble} (except a minor adaptation that we allow self-loops here) by taking $p_{ij}=w_{ij}s$ and $\delta_{ij}=1-s$ for $i,j\in [n]$. The expected degree $d_{i}$ is now equal to $w_{i}s$.  
\end{example}

\Cref{def:graphensemble} generalizes Examples~\ref{eg:correlated er}, \ref{eg:correlated sbm} and \ref{eg:correlated cl} by incorporating the flexibility of vast inhomogeneity for edge probabilities between different pairs of vertices. Our primary contribution in this work is the following theorem. We denote by $\mathfrak S_n$ the collection of all permutations on $[n]$.
\begin{theorem}\label{thm:graphmatching}
For any constant $\alpha \in(0,1)$, there exist constants $c_{1},c_{2},c_{3}$ such that if a pair of random graphs $(G,H)$ follows an $(\alpha,c_{1}\log^{2}n, c_{2}\log^{-2}n)$-edge-correlated distribution, then
$$\PP\left[\mathcal{A}\left(\left(G_{ij}\right)_{i,j\in[n]},\;\left(H_{\pi^{-1}(i),\pi^{-1}(j)}\right)_{i,j\in[n]},\;c_{3}\log n\right)=\pi \text{ for all } \pi \in \mathfrak{S}_{n}\right]\geq 1-O(n^{-1}),$$
where $\mathcal{A}$ stands for \Cref{alg:graphmatching}.
\end{theorem}

In words, \Cref{thm:graphmatching} says that with high probability, our algorithm $\mathcal{A}$ is able to exactly recover the latent vertex correspondence if the minimal correlation is at least $1-\Omega(\log^{-2}n)$ and the minimal average degree is at least $\Omega(\log^{2}n)$. We shall emphasize that our algorithm is based on degree profiles and is rather simple, in the sense that its running time is only cubic. The performance of the algorithm will also be demonstrated in the numerical experiments; see \Cref{sec:numerical}.

\begin{remark}\label{rem-not-access-parameters}
As an appealing feature, our algorithm requires no prior knowledge on the parameters $p_{ij}$ and $\delta_{ij}$. In fact, \Cref{alg:graphmatching} itself does not even rely on the parameters $c_{1},c_{2},c_{3}$ or $\alpha$. But to be clear, the constants $c_{1},c_{2},c_{3}$, which impose probabilistic assumptions on the inputs to \Cref{alg:graphmatching}, are indeed dependent on $\alpha$.
\end{remark}

\subsection{Backgrounds and related works}

The random graph matching problem is motivated by various applied fields such as social network analysis \cite{NS08,NS09}, computer vision \cite{CSS07,BBM05}, computational biology \cite{SXB08,VCP15} and natural language processing \cite{HNM05}, and as a result it is a topic of much interest also from the theoretical point of view. The correlated Erd\H{o}s-R\'{e}nyi graph model, being canonical and simple, has been extensively studied with emphasis placed on the two important and entangling issues---the information threshold (i.e., the statistical threshold) and the computational transition. On the one hand, it is fair to say that thanks to  the collective efforts as in \cite{CK16, CK17, HM23, WXY20+,WXY21+, GML21, DD22+, DD22+b},  we now have a fairly complete understanding on the information thresholds for the problem of correlation detection and vertex matching. On the other hand, while our understanding on the computational aspect remains incomplete, the community has obtained progressively improved efficient algorithms for graph matching; see \cite{PG11, YG13, LFP14, KHG15, FQRM+16, SGE17, BCL19, DMWX21, FMWX22a, FMWX22b, BSH19, CKMP19, DCKG19, MX20, GM20, MRT23, MWXY21+, GMS22+, MWXY22+, DL22+, DL23+}.
The state of the art on algorithms is as follows: in the sparse regime, we have efficient matching algorithms as long as the correlation is above the square root of the Otter's constant (the Otter's constant is around ${0.338}$) \cite{GMS22+, MWXY22+}; in the dense regime, we have efficient matching algorithms as long as the correlation exceeds an arbitrarily small constant \cite{DL23+}. Roughly speaking, the separation of sparse and dense regime above is whether the average degree grows polynomially or sub-polynomially.  In addition, there is a belief that these may be the best possible although currently it seems infeasible to conclusively prove this belief. More realistically, one may try to prove a complexity lower bound under some framework that rules out a (presumably large) class of algorithms. Along this line, efforts have been made on the (closely related) random optimization problem of maximizing the overlap between two independent Erd\H{o}s-R\'{e}nyi graphs; see  \cite{DDG22+, DGH23+}. 

Similar to the random graph matching problem, for many network recovery problems usually the first and the most extensively studied models are of Erd\H{o}s-R\'{e}nyi type; prominent examples include the stochastic block model and the planted clique problem for community recovery. That being said, researchers do realize and keep in mind that almost all the realistic networks, either from a scientific domain or from a real-world application, are drastically different from the Erd\H{o}s-R\'{e}nyi model. Nevertheless, it is a natural hope that a sound theory developed on the Erd\H{o}s-R\'{e}nyi model will eventually be further developed to incorporate more realistic networks. To this end, there are roughly speaking two directions for pursuit: (1) propose and study random graph models that better capture some key features for (at least some) realistic networks; (2) develop \emph{robust} algorithms whose success relies on a minimal set of model assumptions. Some efforts have been made along the first direction: a model for correlated randomly growing graphs was studied in \cite{RS20+}, graph matching for correlated stochastic block model was studied in \cite{RS21},  graph matching for correlated random geometric graphs was studied in \cite{WWXY22+}, and (perhaps most relevant to this work) a correlated \emph{inhomogeneous} Erd\H{o}s-R\'{e}nyi model was studied in \cite{RS23}. Our current work is more along the second direction, and in particular it is distinct from \cite{RS23} not only because we focus on the algorithmic aspect instead of the information threshold as in \cite{RS23}, but also our model assumption seems even more relaxed in the sense that we only specify some bounds on the edge probabilities and correlations instead of working with a set of specific parameters as in \cite{RS23}. Indeed, as pointed out in Remark~\ref{rem-not-access-parameters}, our algorithm does not even need to know these parameters for edge probabilities. We believe that this is an appealing feature since in real applications it is often the case that we do not have access to these parameters, and furthermore in the inhomogeneous case it is also very difficult to obtain good estimates on these edge probabilities (for the obvious reason that the number of parameters is proportional to the number of observed variables).

Let us also remark that the extension to inhomogeneous networks (from Erd\H{o}s-R\'{e}nyi) is usually considered of practical value since many realistic networks exhibit the feature of inhomogeneity, a well-known example being the observed power law degree distribution for social networks. For this reason, extensive works have been done for (e.g.) the stochastic block model in order to incorporate inhomogeneity and degree correction. See \cite{DHM04, COL09, ZLZ12, QR13, TC17, CLX18, GMZZ18, MMY20, NRP21, KJ23} for a partial list.
In light of this, it seems fair to say that our work makes a meaningful step toward the important and ambitious program for \emph{robust} random graph matching algorithms. We acknowledge that our algorithm may be suboptimal in the sense that we believe even under our mild assumptions it may be possible to have efficient matching algorithms when the correlation exceeds a certain (or some arbitrary) constant as in \cite{GMS22+, MWXY22+, DL23+}. That being said, we feel that this suboptimality is to some extent compensated by the fact that our algorithm is rather ``practical'': it is easy to implement and its running-time is cubic. In addition, while it may be challenging to achieve exact matching via degree-profile based algorithms when the correlation is bounded away from 1, it is possible that in some cases (e.g.~\Cref{eg:correlated cl}) inhomogeneity can be leveraged in some degree-profile-based algorithms to obtain a partial recovery and then an exact recovery may be achieved by a seeded algorithm such as in \cite{YXL21}. We leave this for future study.

\subsection{Comparison with earlier degree profile matching algorithm}\label{subsec:intro_comparison}

While our algorithm is hugely inspired by the degree-profile-based matching algorithm proposed in \cite{DMWX21}, we believe our generalization is substantial. In fact, even in the original setting of correlated Erd\H{o}s-R\'{e}nyi graphs, our work seems to make further conceptualization and simplification over \cite{DMWX21}, both in the design and in the analysis of the algorithm. In what follows, we elaborate on this comparison. 

We begin by reviewing the degree-profile algorithm proposed in \cite{DMWX21} (we refer to Sections 1.3 and 1.4 in \cite{DMWX21} for more detailed explanation and for further background). In general, designing a signature-based matching algorithm somehow boils down to choosing the signatures and a distance between two signatures from different graphs. In the degree-profile algorithm in \cite{DMWX21}, the signature for a vertex is chosen as the empirical distribution of the standardized degrees of its neighbours, and the distance is chosen as the $L^{1}$-distance, i.e.~total variation distance, between the appropriately discretized versions of the empirical measures. The discretization leads to a ``balls-into-bins"-type procedure, which measures for each empirical distribution the ``number of balls" (total mass of the distribution) in the ``bins" (the appropriately chosen disjoint intervals), as illustrated below.

\begin{center}
\begin{tikzpicture}
\node[coordinate](0){};
\node[plan](1)[above of=0, node distance = 1cm]{degrees of neighbors of $i$ in $G$};
\node[plan](3)[right of=1, node distance=8cm]{signature of vertex $i$ in $G$};
\node[plan](2)[below of=1, node distance=2cm]{degrees of neighbors of $k$ in $H$};
\node[plan](4)[right of=2, node distance=8cm]{signature of vertex $k$ in $H$};
\node[plan](5)[right of=0, node distance = 12cm]{$D(i,k)$};
\path [line] (1) -- node [text width=5cm,midway,above,align=center] {balls-into-bins} (3);
\path [line] (2) -- node [text width=5cm,midway,above,align=center] {balls-into-bins} (4);
\draw (3)--(9,0);
\draw (4)--(9,0);
\path [line] (9,0) -- node [text width=5cm,midway,above,align=center] {$L^{1}$-distance} (5);
\end{tikzpicture}
\end{center}

The intuition behind the algorithm is simple. If $i\in G$ and $k\in H$ satisfy $k=\pi(i)$, then they share a large number of common vertices whose degrees are correlated random variables, yielding a small $L^{1}$-distance. In contrast, if $k\neq \pi(i)$, then the empirical distributions associated with them are mainly generated from independent random variables, and therefore the $L^{1}$-distance is typically large. We remark here that a cleaner and simpler model with the same spirit is the correlated Gaussian matrix model, and refer to Section 2 of \cite{DMWX21} for analysis and intuitions.

We also remark that on the one hand, the discretization facilitates the analysis (e.g.~in establishing the concentration), while on the other hand, the ``size of bins", i.e.~radii of intervals, need to be chosen carefully so that compared with the standard deviations of the standardized degrees, the bins are neither too crude (which makes fake pairs typically stay in the same bins as true pairs) nor too fine (which makes true pairs typically stay in different bins like fake pairs). In \cite{DMWX21}, the ``size of bins" is of order $\log^{-1}n$, which is the square-root of the noise strength.

We now provide a detailed explanation of our algorithm. Similarly, we also encourage the reader to first read \Cref{sec:gaussian} on correlated Gaussian matrices, which illustrates some key intuitions in a simpler context while hiding various complications under the rug. Briefly speaking, in the Gaussian scenario, we generalize the setting in \cite[Section 2]{DMWX21}, where all Gaussians have mean 0 and variance 1, to the case where the means of Gaussians can be arbitrary. However, the variances are still required to be 1, and it is interesting whether one can further generalize to the inhomogeneous variance situation.

Following the approach of \cite{DMWX21}, in our algorithm the signature will be obtained through a balls-into-bins-type procedure, and we consider the $L^{1}$-distance between signatures (see \Cref{subsec:distance_func} for details). This is primarily in order to facilitate the theoretical analysis of concentration in \Cref{subsec:fakeconcentrate,subsec:trueconcentration}. In particular, as in \cite{DMWX21}, the discretized balls-into-bins-type procedure facilitates the application of McDiarmid's inequality (\Cref{thm:Mcdiarmid}) by simplifying the proof of the bounded-difference property.

\underline{Choice of balls.} The ``balls'' used in \cite{DMWX21} are standardized ``outdegrees" of neighbors. The authors made this choice to reduce dependency and thus to facilitate the analysis. It is reasonable to speculate that using the usual vertex degrees of neighbors would be sufficient for signatures, as empirically verified by the numerical experiments of \cite{DMWX21} in the Erd\H{o}s-R\'{e}nyi case. In contrast to the theoretic result in \cite{DMWX21} but in alignment with its numerical result, our algorithm directly utilizes the degrees of neighbors for the ``balls'', with the exception that we employ a different choice of bins (as described later in this subsection).

Furthermore, by eliminating the standardization step, our algorithm no longer requires prior knowledge on the edge probabilities. However, in the absence of standardization, particularly in the inhomogeneous setting, the degrees of the vertices can vary drastically. This creates difficulties for our choice of bins, which we recall should be neither too crude nor too fine, compared with the standard deviations, for \emph{all} the degrees. Nevertheless, we do know that each vertex's degree is a nonnegative random variable with a standard deviation proportional to the square root of its mean, due to the assumed upper bound $p_{ij}\leq 1-\alpha$. This leads to the following observation on a ``partial'' standardization, which is in essence the classical variance-stabilizing transformation for a sum of independent Bernoulli random variables: taking the square root of a vertex's degree results in a variable with still an arbitrary mean but a standard deviation of constant order. Out of this consideration, we choose the ``balls'' for the signature of $i$ to be the square root of the degrees of neighbors of $i$.

\underline{Choice of bins.} In the paper \cite{DMWX21}, the authors defined the bins $I_{1},\dots,I_{L}$ as the partition of $[-1/2,1/2]$ into equal-length intervals, where $L$ is a sufficiently large multiple of $\log n$. This choice aligns with their use of standardized variables with mean 0 and variance 1 as ``balls.'' In our scenario, the ``balls'' are also of constant-order variances but may have arbitrary means. Conceivably, the natural choice for bins would be a partition of the entire real line into intervals of length $1/L$. In \Cref{sec:gaussian}, we will demonstrate that this selection of bins works well for matching Gaussian matrices, which is a simplified model for graph matching, thanks to the simple and explicit probability density function (pdf). However, such pdf is not available here, and due to certain technical limitations of our analysis method, we will make an artificial choice (see \Cref{subsec:distance_func}) for the bins in our main theorem on graph matching. Specifically, we define the bins as unions of intervals and consider a continuum family of bins. We remark here that taking the union of intervals has the theoretical advantage that, each ball (no matter what the value of its mean is) has at least probability of order $1/L$ to lie in each bin (see \Cref{cor:fakeproblowerb}), avoiding some subtle and complicated estimates. Further justifications on the choice of bins are given below. Note that while we take unions solely for theoretical purposes, having a continuum family of bins does seem to boost the practical performance of the algorithm on specific inhomogeneous models, as suggested by our numerical experiments in \Cref{sec:numerical} (see \Cref{fig:model2,fig:Slashdot Network,fig:oregon} for some particularly evident demonstrations).

\underline{The Fubini-Tonelli theorem.} An important part of the analysis is to upper-bound the probability of a correlated pair of ``balls'' landing in different bins. When a discrete set of bins is employed, this probability estimation poses a technical challenge. In the Gaussian model, as demonstrated in \Cref{lem:CGaussian}, the task is manageable. The complexity increases significantly in the Erd\H{o}s-R\'{e}nyi model, as indicated by \cite[Lemma 9]{DMWX21}. In our setting, which involves inhomogeneity and additional dependence arising from our choice of balls, this task would arguably be even more intricate. It is exactly for the purpose of addressing this challenge that we have chosen a \emph{continuum} family of bins; this then allows us to leverage the Fubini-Tonelli theorem (specifically, in \eqref{eq:T1fubini}), which considerably simplifies this part of analysis (even in the original setting of \cite{DMWX21}).

\underline{Dealing with dependence.} In line with the approach taken in \cite{DMWX21}, we use $L^{2}$-distances to upper-bound the $L^{1}$-distances of correlated pairs of signatures. This requires the bins to be indexed by a set with bounded measure (see \eqref{eq:D'CauchySchwarz}) and forces us to take unions of intervals (or equivalently, cyclic intervals on the circle $\mathbb{R}/\mathbb{Z}$), which are indexed by [0,1), as bins. We will show in \Cref{subsec:upperbonT2}, which we view as one of the main technical contributions of this paper, that the dependence arising from our choice of balls (and our new choice of conditioning, as explained below) can be effectively addressed by transitioning to $L^{2}$-distances.

\underline{Conditioning.} In comparison to \cite{DMWX21}, our analysis places a greater emphasis on the analytical tool of conditioning, bringing it to the forefront. As explained in \Cref{remark:tp1rd,rmk:tp2rd}, the conditioning employed in our analysis is carefully designed to address the additional intricacies inherent in our setting. In particular, we want to condition on less information to preserve the symmetry between $G$ and $H$. A consequence of this approach is that a greater amount of dependence between the ``balls'' will remain after the conditioning process, and this requires a somewhat more complicated analysis for concentration, as a price to pay for easier control on the expectation (from the symmetry).

\subsection{Notations}

In this paper we will use the following convention of terminology for a few times. For a $\sigma$-algebra $\mathcal{F}$ generated by a collection $\boldsymbol{X}$ of random variables and an event $\Gamma\in \mathcal{F}$, we write $\omega \sim (\mathcal {F}, \Gamma)$ if $\omega$ is a realization of $\boldsymbol{X}$ such that $\Gamma$ occurs. 

For a positive integer $n$, we write $[n]=\{1,\dots,n\}$. For any two positive sequences $\{a_n\}$ and $\{b_n\}$, we write equivalently $a_n=O(b_n)$, $b_n = \Omega(a_n)$, $a_n \lesssim b_n$ and $b_n \gtrsim a_n$ if there exists a positive absolute constant $c$ such that $a_n/b_n \leq c$ holds for all sufficiently large $n$. For two real numbers $a$ and $b$, we write $a \vee b:=\max\{a,b\}$.

\subsection{Acknowledgement}
We warmly thank the anonymous referees and the AE for their helpful and detailed comments.
\section{Algorithm for Graph Matching}
In this section, we describe the algorithm for \Cref{thm:graphmatching}. The algorithm takes as input two graphs $\left(G_{ij}\right)_{i,j\in [n]}$ and $\left(H_{ij}\right)_{i,j\in[n]}$, along with a positive integer $L\geq 5$. We first define a distance function $D:[n]\times [n]\rightarrow \mathbb{R}$ based on the inputs.

\subsection{The Distance Function}\label{subsec:distance_func}
For each $t\in [0,1)$ and $m\in \mathbb{Z}$, define the interval 
\begin{equation}\label{eq:defofItm}
I_{t,m}=\left[t+m-\frac{1}{L},t+m+\frac{1}{L}\right],
\end{equation}
and let
\begin{equation}\label{eq:defofIt}
I_{t}=\bigcup_{m\in\mathbb{Z}}I_{t,m}.
\end{equation}
\begin{remark}\label{rmk:cyclicinterval}
We may identify the range $[0,1)$ for $t$ with the circle $\mathbb{R}/\mathbb{Z}$. Fix $t\in [0,1)$. Under the quotient map $\mathbb{R}\rightarrow\mathbb{R}/\mathbb{Z}$, the images of all the intervals $I_{t,m}$ coincide and is a ``cyclic interval'' of length $2/L$. Furthermore, $I_{t}$ is the inverse image of this ``cyclic interval'' under the quotient map. The cyclic-interval viewpoint is of importance in the analysis later (most prominently, in \Cref{lem:rationaleg} on the $L^{2}$-bound, which needs an index set of finite measure, as discussed in \Cref{subsec:intro_comparison}).
\end{remark}
For each $j\in [n]$ we define two random variables
$$X_{j}:=\sum_{\ell=1}^{n}G_{j\ell}\quad \text{and} \quad Y_{j}:=\sum_{\ell=1}^{n}H_{j\ell}.$$
The variables $X_{j}$ and $Y_{j}$ are the degrees of $G$ and $H$ respectively. For each $i,j\in[n]$ and $t\in [0,1)$ consider the pair of Bernoulli random variables
\begin{equation}\label{eq:defofA}
A_{ij}(t)=\mathbbm{1}\left\{\sqrt{X_{j}}\in I_{t}\right\}\cdot\mathbbm{1}\{G_{ij}=1\},
\end{equation}
\begin{equation}\label{eq:defofB}
B_{ij}(t)=\mathbbm{1}\left\{\sqrt{Y_{j}}\in I_{t}\right\}\cdot\mathbbm{1}\{H_{ij}=1\}.
\end{equation}
Now we can define the distance function
\begin{equation}\label{eq:defofdistance}
D(i,k)=\int_{0}^{1}\left|\sum_{j=1}^{n}A_{ij}(t)-\sum_{j=1}^{n}B_{kj}(t)\right|\d t. 
\end{equation}

Fix an $i \in [n]$. The sets $(I_{t})_{0\leq t<1}$ can be seen as a continuum family of ``bins'', and the numbers $\sqrt{X_{j}}$ (for $j \in [n]$) can be seen as ``balls''. A ball $\sqrt{X_{j}}$ is cast into (an uncountable number of) bins only if $j$ is a neighbor of $i$ in $G$. We then investigate $\sum_{j=1}^{n}A_{ij}(t)$ for $t \in [0,1)$, that is, the number of balls in each bin $I_{t}$. These numbers form the ``signature'' of the vertex $i$ in $G$. Signatures of vertices in $H$ are similarly obtained, and the quantity $D(i,k)$ measures the $L^{1}$-distance between the signature of $i$ in $G$ and the signature of $k$ in $H$. We refer to \Cref{subsec:intro_comparison} for a detailed explanation on the ideas behind this balls-into-bins-type procedure and our choice of the balls and the bins. In particular, taking the square root serves as a variance-stabilizing transformation, which is important under our vastly inhomogeneous setting.

By definitions in \eqref{eq:defofItm}, \eqref{eq:defofIt} and \eqref{eq:defofA}, it is easy to observe that the (random) function $A_{ij}(\cdot)$ is either a zero function or an indicator function of a ``cyclic interval'' of length $2/L$ centered at $\sqrt{X_{j}}$. We state this fact as a formal proposition for convenience of future reference.

\begin{proposition}\label{prop:structureofA}
For any $j \in [n]$, the function $A_{ij}(\cdot)$ is the indicator function of
$$
\begin{cases}
\emptyset &\text{if }G_{ij}=0,\\
\Big[\{\sqrt{X_{j}}-\frac{1}{L}\},\;\{\sqrt{X_{j}}+\frac{1}{L}\}\Big] &\text{if }G_{ij}=1\text{ and }\frac{1}{L}\leq \{\sqrt{X_{j}}\}<1-\frac{1}{L},\\
\Big[0,\;\{\sqrt{X_{j}}+\frac{1}{L}\}\Big]\cup\Big[\{\sqrt{X_{j}}-\frac{1}{L}\},\;1\Big) &\text{otherwise},
\end{cases}$$
where $\{x\}:=x-\lfloor x\rfloor$ stands for the fractional part of the real number $x$. In particular, the support of $A_{ij}(\cdot)$ has Lebesgue measure either $0$ or $2/L$.
\end{proposition}

\Cref{prop:structureofA} and its obvious analogue for $B_{ij}(\cdot)$ will be used many times in the analysis later (most prominently, in \Cref{lem:T1bound}).
\subsection{The Algorithm}
We are now ready to describe the main algorithm:

\begin{algorithm}[H]
\DontPrintSemicolon
\SetKwInOut{Input}{Input}\SetKwInOut{Output}{Output}
\caption{Random Graph Matching}\label{alg:graphmatching}
\Input{Graphs $(G_{ij})_{i,j\in [n]}$ and $(H_{ij})_{i,j\in [n]}$, and an integer $L$}
\Output{A permutation $\widehat{\pi}\in \mathfrak{S}_{n}$, or ``error''}
For $i,k\in[n]$ compute $D(i,k)$ according to \eqref{eq:defofdistance}\;
\For{each $i\in [n]$}{
    Let $\widehat{\pi}(i)=\argmin_{k\in [n]}D(i,k)$ (output \textbf{Error} if there are multiple minimums) 
}
\eIf{$\widehat{\pi}$ is a permutation on $[n]$}{
    Output $\widehat{\pi}$
}{
    Output \textbf{Error}
}
\end{algorithm}

\subsection{Time Complexity}

In this subsection, we show that the running time of \Cref{alg:graphmatching} is $O(n^{3})$. In \Cref{alg:graphmatching}, the total time needed to execute line 2 through line 7 is no more than $O(n^{2})$. We then focus on the time complexity of carrying out line 1. 

To compute the distance table $[D(i,k)]_{i,k\in[n]}$, we first compute the degrees of each vertices. Whether the graphs $G$ and $H$ are represented using adjacency lists or adjacency matrices, computing the degrees $(X_{i})_{i\in [n]}$ and $(Y_{i})_{i\in [n]}$ takes no more than $O(n^{2})$ time. 

We then compute the landmarks $\{\sqrt{X_{j}}-\frac{1}{L}\}$, $\{\sqrt{X_{j}}+\frac{1}{L}\}$, $\{\sqrt{Y_{j}}-\frac{1}{L}\}$ and $\{\sqrt{Y_{j}}+\frac{1}{L}\}$, for $j\in [n]$. We sort the $4n$ landmarks as $a_{1}\leq a_{2}\leq \cdots\leq a_{4n}$. Computing and sorting the landmarks take $O(n\log n)$ time in total.

Let $a_{0}=0$ and $a_{4n+1}=1$. By \Cref{prop:structureofA}, for any $i\in[n]$, the function $\sum_{j=1}^{n}A_{ij}(\cdot)$ is constant on each interval $(a_{\ell},a_{\ell+1})$. Let $u_{i\ell}$ be the value of $\sum_{j=1}^{n}A_{ij}(\cdot)$ on $(a_{\ell},a_{\ell+1})$, for $\ell\in\{0,1,\dots,4n\}$. Then the signature of vertex $i$ in $G$ is characterized by the values $(u_{i\ell})_{0\leq \ell\leq 4n}$. For each vertex $i\in[n]$ and index $\ell\in\{0,1,\dots,4n-1\}$, it takes $O(1)$ time to compute $u_{i, \ell+1}$ given $u_{i\ell}$. So it takes $O(n)$ time to (progressively) compute the values $(u_{i\ell})_{0\leq \ell\leq 4n}$ . In total, we need $O(n^{2})$ time to compute $u_{i\ell}$ for all $i \in [n]$ and $0 \leq \ell \leq 4n$. Analogously define the vectors $(v_{k\ell})_{0\leq \ell\leq 4n}$ for signatures in $H$. By symmetry, computing $v_{k\ell}$ for all $k \in [n]$ and $0 \leq \ell \leq 4n$ also needs $O(n^{2})$ time.

Now, for each pair $i,k\in [n]$, we have
$$D(i,k)=\sum_{\ell=0}^{4n}(a_{\ell+1}-a_{\ell})\left|v_{k\ell}-u_{i\ell}\right|,$$
which requires at most $O(n)$ time to compute. In total, we need $O(n^{3})$ time to compute the full table $[D(i,k)]_{i,k\in[n]}$.
\subsection{Pre-Analysis of the Algorithm}\label{subsec:preanalysis}

It is easy to observe that \Cref{alg:graphmatching} is permutation-oblivious in the following sense.

\begin{proposition}\label{prop:algoblivious}
For any pair of graphs $(G,H)$ and any integer $L$, we always have
$$\mathcal{A}\left(\left(G_{ij}\right)_{i,j\in[n]},\;\left(H_{\pi^{-1}(i),\pi^{-1}(j)}\right)_{i,j\in[n]},\; L\right)=\pi\circ \mathcal{A}\left(\left(G_{ij}\right)_{i,j\in[n]},\;\left(H_{ij}\right)_{i,j\in[n]},\; L\right),$$
where $\mathcal{A}$ stands for \Cref{alg:graphmatching}.
\end{proposition}

Throughout the analysis (\Cref{subsec:preanalysis}, \Cref{sec:fakepair} and \Cref{sec:truepair}), we assume the true permutation $\pi$ is the identity $\mathrm{id}$, which is without loss of generality due to \Cref{prop:algoblivious}. In other words, we assume that the input $(G,H)$ follows an $(\alpha,d,\delta)$-edge-correlated distribution defined in \Cref{def:graphensemble}, without having their vertices permuted.

We then state the following two main lemmas.
\begin{lemma}[Fake pair]\label{lem:fakepair}
For any constant $\alpha\in(0,1)$, there exist constants $c_{1},c_{3}$ such that setting $L=c_{3}\log n$, if a pair of random graphs $(G,H)$ follows an $(\alpha,c_{1}\log^{2}n,\delta)$-edge-correlated distribution, then for any pair of distinct indices $i,k \in [n]$, 
$$\PP\left[\frac{D(i,k)}{\sqrt{d_{i}}}\leq \frac{200}{c_{3}\sqrt{\log n}}\right]\leq O(n^{-3}),$$
where we recall that $d_{i}=\sum_{j=1}^{n}p_{ij}$ is the degree parameter for the vertex $i$, as defined in \Cref{def:graphensemble}.
\end{lemma}
\begin{lemma}[True pair]\label{lem:truepair}
For any constant $\alpha\in(0,1)$ and $c_{3}$, there exist constants $c_{1}',c_{5},c_{6}$ such that setting $L=c_{3}\log n$, if a pair of random graphs $(G,H)$ follows an  $(\alpha,c_{1}'\log^{2}n,\delta)$-edge-correlated distribution, then for any index $i\in[n]$, 
$$\PP\left[\frac{D(i,i)}{\sqrt{d_{i}}}\geq c_{5}\sqrt{\delta\log n}+c_{6}\delta^{1/4}+\frac{100}{c_{3}\sqrt{\log n}}\right]\leq O(n^{-3}).$$
\end{lemma}
\Cref{thm:graphmatching} then follows as a consequence of the previous two lemmas.
\begin{proof}[Proof of \Cref{thm:graphmatching}]
We first apply \Cref{lem:fakepair} and obtain constants $c_{1}$ and $c_{3}$. Then apply \Cref{lem:truepair} to the fixed constant $c_{3}$ and obtain constants $c_{1}',c_{5}$ and $c_{6}$. Finally, pick positive constant $c_{2}$ sufficiently small such that
$$c_{5}\sqrt{c_{2}}+c_{6}c_{2}^{1/4}+\frac{100}{c_{3}}<\frac{200}{c_{3}}.$$
It now follows from \Cref{lem:fakepair,lem:truepair} that if a pair of random graphs $(G,H)$ follows an $\left(\alpha,\allowbreak \max\{c_{1},c_{1}'\}\log^{2}n, c_{2}\log^{-2}n\right)$-edge-correlated distribution, then for any pair of distinct indices $i,k \in [n]$,
$$\PP\left[D(i,i)\geq D(i,k)\right]\leq O(n^{-3}).$$
As explained at the beginning of \Cref{subsec:preanalysis}, we assume that the underlying permutation $\pi=\mathrm{id}$. Note that we have
$$\PP\left[\mathcal{A}(G,H,c_{3}\log n)\neq \mathrm{id}\right]\leq \PP\left[\exists i,k\in [n]\text{ s.t. }i\neq k\text{ and }D(i,i)\geq D(i,k)\right]\leq O(n^{-1}),$$
by a union bound over all pairs $i,k$. This finishes the proof of \Cref{thm:graphmatching}.
\end{proof}

\section{Fake Pairs}\label{sec:fakepair}

This section is devoted to proving \Cref{lem:fakepair}. Our strategy for proving \Cref{lem:fakepair} is to do two rounds of conditioning on the original randomness for $(G,H)$ before carrying out the main part of the analysis. We will define the conditionings in \Cref{subsec:fakefirstcond,subsec:fakeseccond} and then analyze the remaining randomness in \Cref{subsec:fakeconcentrate,subsec:fakeexpectation}. 

In the following, let $(G,H)$ be a pair of random graphs that follows an $(\alpha,d,\delta)$-edge-correlated distribution, and we fix vertices $i\neq k$. Without loss of generality, we assume $d_{i}\geq d_{k}$. We remark that, in the remaining of this work, the superscript $\mathsf{f}$ refers to ``fake'', and is not a parameter.
\subsection{First Round of Conditioning}\label{subsec:fakefirstcond}

\begin{definition}[Variables]
For each $j\in [n]$, define the random variable
$$R_{j}^{\fake}=\sum_{\ell\not\in\{i,k\}}p_{j\ell}(1-G_{i\ell}\vee H_{k\ell}).$$
Also define the random variables
$$R_{0}^{\fake}=\sum_{j\neq k}G_{ij}(1-H_{kj})\text{ and }R^{\fake}=2+\sum_{j\not\in\{i,k\}}G_{ij}\vee H_{kj}.$$
\end{definition}
\begin{definition}[Events]
Let $\Gamma^{\fake}$ be the event
$$\Gamma^{\fake}:=\bigcap_{j=1}^{n}\left\{R_{j}^{\fake}\geq \frac{\alpha^{2}}{2}d_{j}\right\}\cap\left\{R_{0}^{\fake}\geq \frac{\alpha}{2}d_{i}\right\}\cap\left\{R^{\fake}\leq 3d_{i}\right\}.$$
\end{definition}
\begin{proposition}\label{prop:fakeGamma}
If $d\gtrsim \log^{2}n$, then $\PP(\Gamma^{\fake})\geq 1-\exp(-\Omega(\log^{2}n))$.
\end{proposition}
\begin{proof}
Note that the $2(n-2)$ random variables $G_{i\ell}, H_{k\ell}$ for $\ell\not\in \{i,k\}$ are mutually independent. So 
$$\E[R_{j}^{\fake}]\geq \sum_{\ell\not\in\{i,k\}}p_{j\ell}\cdot \alpha^{2}\geq \alpha^{2}(d_{j}-2).$$
By Hoeffding's inequality (\Cref{thm:Hoeffding}) and $p_{j\ell} \leq 1$ for all $j,\ell$, 
$$\PP\left[R_{j}^{\fake}\leq \frac{\alpha^{2}}{2}d_{j}\right]\leq 2\exp\left(-\frac{2\left(\alpha^{2}(d_{j}-2)-\frac{1}{2}\alpha^{2}d_{j}\right)^{2}}{\sum_{\ell\not\in\{i,k\}}p_{j\ell}^{2}}\right)\leq 2\exp\left(-\frac{\alpha^{4}}{2}(d_{j}-8)\right).$$
For $R_{0}^{\fake}$ and $R^{\fake}$, we have by independence
$$\E[R_{0}^{\fake}]=\sum_{j\neq k}p_{ij}(1-p_{kj})\geq \alpha(d_{i}-1)\text{ and }\E[R^{\fake}]\leq d_{i}+d_{k}+2\leq 2d_{i}+2.$$
Then we have by Chernoff bounds (\Cref{cor:Chernoff}) that 
$$\PP\left[R_{0}^{\fake}\leq \frac{\alpha}{2}d_{i}\right]\leq \exp\left(-\frac{\alpha(d_{i}/2-2)}{4}\right)= \exp\left(-\frac{\alpha}{8}(d_{i}-4)\right),$$
and that 
$$\PP[R^{\fake}\geq 3d_{i}]\leq  \exp\left(-\frac{1}{6}(d_{i}-5)\right).$$
It now follows from a union bound that $\PP(\Gamma^{\fake})\geq 1-\exp(-\Omega(\log^{2}n))$.
\end{proof}
\begin{definition}[Conditioning]\label{def:fp1rd}
Let $\mathcal{F}^{\fake}$ be the $\sigma$-algebra generated by the variables $G_{ij}$ and $H_{kj}$ for all $j\in[n]$.
\end{definition} 
Clearly, $\Gamma^{\fake}$ is $\mathcal{F}^{\fake}$-measurable. Throughout the remaining part of \Cref{sec:fakepair}, we will always condition on $\mathcal{F}^{\fake}$ and assume that $\Gamma^{\fake}$ holds.

\subsection{Second Round of Conditioning}\label{subsec:fakeseccond}

\begin{definition}[Sets]
Define the sets
$$J^{\fake}=\{j\in[n]\setminus\{i,k\}:G_{ij}\vee H_{kj}=1\}$$ and $$J_{1}^{\fake}=\{j\in [n]\setminus\{i,k\}:G_{ij}=1,H_{kj}=0\}\subset J^{\fake}.$$
\end{definition}
Note that under the conditioning of $\mathcal{F}^{\fake}$, the sets $J^{\fake}$ and $J_{1}^{\fake}$ are deterministic.
\begin{definition}[Variables]
For $j\in J_{1}^{\fake}$, define the random variable 
$$S_{j}^{\fake}=\sum_{\ell \in J^{\fake}\cup\{i,k\}}G_{j\ell}.$$ 
\end{definition}
\begin{definition}[Events]
Let $\Pi^{\fake}$ be the event
$$\Pi^{\fake}:=\bigcap_{j\in J_{1}^{\fake}}\{S_{j}^{\fake}\leq 2d_{j}\}.$$
\end{definition}
\begin{proposition}\label{prop:fakeTheta}
If $d\gtrsim\log^{2}n$, then $\PP(\Pi^{\fake}|\mathcal{F}^{\fake})\geq 1-\exp(-\Omega(\log^{2}n))$.
\end{proposition}
\begin{proof}
By definition of $S_{j}^{\fake}$, we have 
$$\E[S_{j}^{\fake}|\mathcal{F}^{\fake}]\leq 2+\sum_{\ell\in J^{\fake}}p_{j\ell}\leq d_{j}+2.$$
Since $G_{j\ell}$ (for $\ell\in J^{\fake}\cup\{i,k\}$) are still independent variables after conditioning on $\mathcal{F}^{\fake}$, by Chernoff bounds (\Cref{cor:Chernoff}) we have
$$\PP[S_{j}^{\fake}\geq 2d_{j}|\mathcal{F}^{\fake}]\leq \exp\left(-\frac{d_{j}-6}{3}\right).$$
It follows from a union bound that $\PP(\Pi^{\fake}|\mathcal{F}^{\fake})\geq 1-\exp(-\Omega(\log^{2}n))$.
\end{proof}
\begin{definition}\label{def:fp2rd}
Let $\mathcal{G}^{\fake}$ be the $\sigma$-algebra generated by $\mathcal{F}^{\fake}$ and the variables $G_{j\ell},H_{j\ell}$ for $j,\ell\in J^{\fake}\cup\{i,k\}$.
\end{definition}
Clearly, $\Pi^{\fake}$ is $\mathcal{G}^{\fake}$-measurable. Throughout the rest of \Cref{sec:fakepair}, we will always condition on $\mathcal{G}^{\fake}$ and assume that $\Gamma^{\fake}\cap\Pi^{\fake}$ holds.

\subsection{Concentration of Distance}\label{subsec:fakeconcentrate}

One important way in which the conditional distribution of $D(i,k)$ given $\mathcal{G}^{\fake}$ behaves nicer than the original distribution of $D(i,k)$ is that its concentration is more directly analyzable after the conditioning. In particular, the randomness in $D(i,k)$ now naturally decomposes into independent parts, lending itself to the powerful tool of McDiarmid's inequality.  

\begin{lemma}[Concentration]\label{lem:fakeconcentration}
On $\Gamma^{\fake}\cap\Pi^{\fake}$, we have
$$\PP\left[D(i,k)\leq \E[D(i,k)|\mathcal{G}^{\fake}]-\frac{100\sqrt{d_{i}\log n}}{L}\middle| \mathcal{G}^{\fake}\right]\leq O(n^{-3}).$$
\end{lemma}
\begin{proof}
After both rounds of conditioning, the distance $D(i,k)$ is now a function of the vectors $(X_{j},Y_{j})$, for $j\in J^{\fake}\cup\{i,k\}$. In addition, these vectors are now independent of each other, as $(X_{j},Y_{j})$ now only depends on $(G_{j\ell},H_{j\ell})$ for $\ell\not\in J^{\fake}\cup\{i,k\}$. 

Now, fix some $j_{0}\in J^{\fake}\cup\{i,k\}$ and consider the marginal difference the vector $(X_{j_{0}},Y_{j_{0}})$ can cause on the value of $D(i,k)$ when the values of all other vectors are fixed. As defined in \eqref{eq:defofdistance}, the value of $D(i,k)$ is determined by the $2n$ functions $A_{ir}(\cdot)$ and $B_{kr}(\cdot)$, for $r\in [n]$. Keeping $(X_{j},Y_{j})$ fixed for all $j\in J^{\fake}\cup\{i,k\}\setminus\{j_{0}\}$, a change in the vector $(X_{j_{0}},Y_{j_{0}})$ will only change two of the $2n$ functions: $A_{ij_{0}}(t)$ and $B_{kj_{0}}(t)$. From \Cref{prop:structureofA} we know that $A_{ij_{0}}(\cdot)$ and $B_{kj_{0}}(\cdot)$ are always indicator functions with supports having Lebesgue measure no more than $2/L$. So using the superscript $\mathsf{new}$ to denote the values after the change in the vector $(X_{j_{0}},Y_{j_{0}})$, we have
\begin{align*}
\left|D(i,k)^{\mathsf{new}}-D(i,k)\right|&\leq \sum_{r=1}^{n}\int_{0}^{1}\Big(\left|A_{ir}(t)^{\mathsf{new}}-A_{ir}(t)\right|+\left|B_{kr}(t)^{\mathsf{new}}-B_{kr}(t)\right|\Big)\d t\\
&\leq\int_{0}^{1}\left|A_{ij_{0}}(t)^{\mathsf{new}}\right|\d t+\int_{0}^{1}\left|A_{ij_{0}}(t)\right|\d t+\int_{0}^{1}\left|B_{kj_{0}}(t)^{\mathsf{new}}\right|\d t+\int_{0}^{1}\left|B_{kj_{0}}(t)\right|\d t\\
&\leq \frac{2}{L}+\frac{2}{L}+\frac{2}{L}+\frac{2}{L}=\frac{8}{L}.
\end{align*}

At any $\omega\sim(\mathcal{G}^{\fake},\Gamma^{\fake}\cap\Pi^{\fake})$, we can now conclude by McDiarmid's inequality (\Cref{thm:Mcdiarmid}) that with conditional probability (on $\mathcal{G}^{\fake}$) at least $1-O\left(n^{-3}\right)$,
$$D(i,k)\geq \E[D(i,k)|\mathcal{G}^{\fake}]-\sqrt{\frac{1}{2}\cdot3\log n\cdot |J^{\fake}\cup\{i,k\}|\cdot\left(\frac{8}{L}\right)^{2}}\geq \E[D(i,k)|\mathcal{G}^{\fake}]-\frac{100\sqrt{d_{i}\log n}}{L}.$$
In the last inequality we used $|J^{\fake}\cup\{i,k\}|\leq R^{\fake}\leq 3d_{i}$, as guaranteed by $\Gamma^{\fake}$.
\end{proof}

\subsection{Expectation of Distance}\label{subsec:fakeexpectation}

Having \Cref{lem:fakeconcentration} at our disposal, the next logical step towards a proof of \Cref{lem:fakepair} is to give a lower bound on the conditional expectation $\E[D(i,k)| \mathcal{G}^{\fake}]$ evaluated on $\Gamma^{\fake}\cap\Pi^{\fake}$. This is what we will achieve in this subsection.

Recall that conditional on $\mathcal{G}^{\fake}$, the set $J_{1}^{\fake}$ is deterministic. Fix $j \in J_{1}^{\fake}$. On the one hand, we get from the  definition of $J_{1}^{\fake}$ that conditional on $\mathcal{G}^{\fake}$, almost surely $Y_{j}=0$. On the other hand, $X_{j}$ can now be written as
$$X_{j}=S_{j}^{\fake}+\sum_{\ell\not\in J^{\fake}\cup\{i,k\}}G_{j\ell}.$$
Conditional on $\mathcal{G}^{\fake}$, the first part $S_{j}^{\fake}$ is a deterministic quantity, while the second part is a sum of independent Bernoulli variables. We have
$$\E[X_{j}|\mathcal{G}^{\fake}]=S_{j}^{\fake}+R_{j}^{\fake}$$
and
$$\mathrm{Var}[X_{j}|\mathcal{G}^{\fake}]=\sum_{\ell\not\in J^{\fake}\cup\{i,k\}}p_{j\ell}(1-p_{j\ell})\in[\alpha R_{j}^{\fake},R_{j}^{\fake}].$$
In addition, recall that on $\Gamma^{\fake}\cap\Pi^{\fake}$ we have $S_{j}^{\fake}\leq 2d_{j}$ and $\alpha^{2}d_{j}/2\leq R_{j}^{\fake}\leq d_{j}$.
\begin{lemma}\label{lem:fakeproblowerb}
Fix some $\omega\sim (\mathcal{G}^{\fake},\Gamma^{\fake}\cap\Pi^{\fake})$ of the underlying probability space.
Fix a vertex $j\in J_{1}^{\fake}$. Let $a,b$ be positive real numbers such that $b-a=2/L$ and 
$$\sqrt{S_{j}^{\fake}+R_{j}^{\fake}}-2\leq a<b\leq \sqrt{S_{j}^{\fake}+R_{j}^{\fake}}.$$
Let 
\begin{equation}\label{eq:gamma_j}
\gamma_{j}:=\frac{\sqrt{2/\pi}\cdot \alpha e^{-100/\alpha^{2}}}{L}-\frac{2\sqrt{2}\alpha^{-3/2}}{\sqrt{d_{j}}}.
\end{equation}
Then we have at $\omega$
$$\PP\left[a\leq \sqrt{X_{j}}\leq b\middle|\mathcal{G}^{\fake}\right]\geq \gamma_{j}.$$
\end{lemma}
\begin{proof}
Let $X'$ be a normal random variable with mean $S_{j}^{\fake}+R_{j}^{\fake}=\E[X_{j}|\mathcal{G}^{\fake}]$ and variance $\mathrm{Var}[X_{j}|\mathcal{G}^{\fake}]$. Then
\begin{align*}
\left|a^{2}-(S_{j}^{\fake}+R_{j}^{\fake})\right|&= \left|a-\sqrt{S_{j}^{\fake}+R_{j}^{\fake}}\right|\cdot \left|a+\sqrt{S_{j}^{\fake}+R_{j}^{\fake}}\right|\\
&\leq 2\cdot 2\sqrt{S_{j}^{\fake}+R_{j}^{\fake}}\\
&\leq 4\sqrt{3d_{j}}.
\end{align*}
Similarly $|b^{2}-(S_{j}^{\fake}+R_{j}^{\fake})|\leq 4\sqrt{3d_{j}}$. So the density function of $X'$ on the interval $[a^{2},b^{2}]$ is no less than
$$\frac{1}{\sqrt{2\pi\mathrm{Var}[X']}}\exp\left(-\frac{(4\sqrt{3d_{j}})^{2}}{2\mathrm{Var}[X']}\right)\geq\frac{1}{\sqrt{2\pi d_{j}}}e^{-100/\alpha^{3}}.$$
Since $b^{2}-a^{2}\geq (b-a)\left(2\sqrt{S_{j}^{\fake}+R_{j}^{\fake}}-4\right)\geq 2\alpha\sqrt{d_{j}}/L$, it follows that
$$\PP[a^{2}\leq X'\leq b^{2}]\geq \sqrt{2/\pi}\cdot \alpha e^{-100/\alpha^{2}}\cdot \frac{1}{L}.$$
By Berry-Esseen theorem (\Cref{cor:B-E}), we have at $\omega$
\[\PP[a^{2}\leq X_{j}\leq b^{2}|\mathcal{G}^{\fake}]\geq \PP[a^{2}\leq X'\leq b^{2}]-\frac{2}{\sqrt{\mathrm{Var}[X_{j}|\mathcal{G}^{\fake}]}}\geq \frac{\sqrt{2/\pi}\cdot \alpha e^{-100/\alpha^{2}}}{L}-\frac{2\sqrt{2}\alpha^{-3/2}}{\sqrt{d_{j}}}=\gamma_{j}.\qedhere\]
\end{proof}
\begin{corollary}\label{cor:fakeproblowerb}
Assume that $L\geq 5$. Fix a realization $\omega\sim (\mathcal{G}^{\fake},\Gamma^{\fake}\cap\Pi^{\fake})$, and fix a vertex $j\in J_{1}^{\fake}$. Let $\gamma_{j}$ be as in \eqref{eq:gamma_j}. Then for each $t\in [0,1)$, we have at $\omega$
\begin{equation}\label{eq:prob_Bernoulli}
\gamma_{j} \leq \PP[A_{ij}(t)=1|\mathcal{G}^{\fake}]\leq 1-\gamma_{j}.
\end{equation}
\end{corollary}
\begin{proof}
For each $t\in [0,1)$, there is some $m\in\mathbb{Z}$ such that $$\sqrt{S_{j}^{\fake}+R_{j}^{\fake}}-\frac{3}{2}\leq t+m\leq \sqrt{S_{j}^{\fake}+R_{j}^{\fake}}-\frac{1}{2}.$$
So applying \Cref{lem:fakeproblowerb} with $a=t+m-1/L$ and $b=t+m+1/L$, we have
\begin{equation}\label{eq:lower_bound_prob_Itm}
\PP\left[\sqrt{X_{j}}\in I_{t,m}\middle|\mathcal{G}^{\fake}\right]\geq \gamma_{j}.
\end{equation}

Now, since $j\in J_{1}^{\fake}$, it is guaranteed that $G_{ij}=1$, and hence $A_{ij}(t)=\mathbbm{1}\left\{\sqrt{X_{j}}\in I_{t}\right\}$. So 
$$\PP[A_{ij}(t)=1 |\mathcal{G}^{\fake}]=\PP\left[\sqrt{X_{j}}\in I_{t}\middle|\mathcal{G}^{\fake}\right]\geq \PP\left[\sqrt{X_{j}}\in I_{t,m}\middle|\mathcal{G}^{\fake}\right],$$
and the lower bound in \eqref{eq:prob_Bernoulli} follows by combining the above with \eqref{eq:lower_bound_prob_Itm}.

For the upper bound, let $t'\in [0,1)$ be such that $|t'-t|=\frac{1}{2}$. Since $\frac{1}{L}<\frac{1}{4}$, we have $I_{t}\cap I_{t'}=\emptyset$. So the events $\left\{\sqrt{X_{j}}\in I_{t}\right\}$ and $\left\{\sqrt{X_{j}}\in I_{t'}\right\}$ are mutually exclusive, and thus
$$\PP[A_{ij}(t)=1 |\mathcal{G}^{\fake}]= \PP\left[\sqrt{X_{j}}\in I_{t}\middle|\mathcal{G}^{\fake}\right]\leq 1-\PP\left[\sqrt{X_{j}}\in I_{t'}\middle|\mathcal{G}^{\fake}\right]= 1-\PP[A_{ij}(t')=1 |\mathcal{G}^{\fake}].$$
The upper bound in \eqref{eq:prob_Bernoulli} then follows by applying the lower bound to $t'$. 
\end{proof}
\begin{lemma}[Expectation]\label{lem:fakeexpectation}
For any $\alpha\in(0,1)$, there exist constants $c_{1},c_{3}$ such that setting $L=c_{3}\log n$ and $d=c_{1}\log^{2}n$, on $\Gamma^{\fake}\cap \Pi^{\fake}$ we have
$$\E[D(i,k)|\mathcal{G}^{\fake}]\geq \frac{300\sqrt{d_{i}}}{c_{3}\sqrt{\log n}}.$$
\end{lemma}
\begin{proof}
For $j \in J_{1}^{\fake}$, let $\gamma_{j}$ be as in \eqref{eq:gamma_j}. Since $L=c_{3}\log n$ and $d_{j}\geq d=c_{1}\log^{2}n$ for each $j$, by taking a sufficiently large $c_{3}$ and then a sufficiently large $c_{1}$, we can ensure that 
$$\min_{j \in J_{1}^{\fake}}\gamma_{j} \geq \frac{10^{7}}{\alpha c_{3}^{2}\log n}.$$
In addition, since on $\Gamma^{\fake}$ we have $|J_{1}^{\fake}|=R_{0}^{\fake}\geq \alpha d_{i}/2\geq \frac{1}{2}\alpha c_{1}\log ^{2}n$, we can also ensure that $\min_{j\in J_{1}^{\fake}}\gamma_{j}\geq |J_{1}^{\fake}|^{-1}$.
Note that conditional on $\mathcal{G}^{\fake}$, the random vectors $(A_{ij},B_{kj})$ (for $j\in [n]$) are mutually independent.  
Combining \Cref{cor:fakeproblowerb} with \Cref{cor:control_f}, on $\Gamma^{\fake}\cap\Pi^{\fake}$ we have 
$$\E\left[\left|\sum_{j\in J_{1}^{\fake}}A_{ij}(t)+\sum_{j\not\in J_{1}^{\fake}}(A_{ij}(t)-B_{kj}(t))\right|\mid\mathcal{G}^{\fake}\right]\geq \frac{1}{4}\sqrt{\sum_{j \in J_{1}^{\fake}}\gamma_{j}}\geq \frac{1}{4}\sqrt{\frac{\alpha d_{i}}{2}\cdot\frac{10^{7}}{\alpha c_{3}^{2}\log n}}\geq \frac{300\sqrt{d_{i}}}{c_{3}\sqrt{\log n}}.$$
for each $t\in[0,1)$. By definition of 
$J_{1}^{\fake}$, for each $j\in J_{1}^{\fake}$ we must have $B_{kj}(t)=0$ for all $t$. So on $\Gamma^{\fake}\cap\Pi^{\fake}$ we have
\[\E[D(i,k)|\mathcal{G}^{\fake}] =\int_{0}^{1}\E\left[\left|\sum_{j\in J_{1}^{\fake}}A_{ij}(t)+\sum_{j\not\in J_{1}^{\fake}}(A_{ij}(t)-B_{kj}(t))\right|\mid \mathcal{G}^{\fake}\right]\d t\geq \frac{300\sqrt{d_{i}}}{c_{3}\sqrt{\log n}}.\qedhere\]
\end{proof}

\subsection{Proof of \Cref{lem:fakepair}}\label{subsec:fakeproof}

We now combine the concentration results of \Cref{subsec:fakeconcentrate} and the expectation results of \Cref{subsec:fakeexpectation} to prove the deviation probability stated in \Cref{lem:fakepair}.

\begin{proof}[Proof of \Cref{lem:fakepair}]
First pick $c_{1}$ and $c_{3}$ according to \Cref{lem:fakeexpectation}. Then \Cref{lem:fakeconcentration,lem:fakeexpectation} together imply that on $\Gamma^{\fake}\cap\Pi^{\fake}$,
$$\PP\left[\frac{D(i,k)}{\sqrt{d_{i}}}\leq \frac{200}{c_{3}\sqrt{\log n}}\middle|\mathcal{G}^{\fake}\right]\leq O(n^{-3}).$$
Finally, \Cref{prop:fakeGamma,prop:fakeTheta} imply that $\PP(\Gamma^{\fake}\cap\Pi^{\fake})\geq 1-\exp(-\Omega(\log^{2}n))$.
Therefore
\[\PP\left[\frac{D(i,k)}{\sqrt{d_{i}}}\leq \frac{200}{c_{3}\sqrt{\log n}}\right]\leq O(n^{-3})+\exp(-\Omega(\log^{2}n))=O(n^{-3}).\qedhere\]
\end{proof}
\section{True Pairs}\label{sec:truepair}

This section is devoted to proving \Cref{lem:truepair}. In order to analyze true pairs, we also need to do two rounds of conditioning, which will be given in \Cref{subsec:truefirstcond,subsec:trueseccond}. We will then analyze the remaining randomness in \Cref{subsec:truncation,subsec:trueconcentration,subsec:trueexpectation,subsec:upperbonT2}. 

In the following, let $(G,H)$ be a pair of random graphs that follows an $(\alpha,d,\delta)$-edge-correlated distribution, where $\alpha$ is a fixed constant. We fix a vertex $i\in [n]$. We remark that, in the remaining of this work, the superscript $\mathsf{t}$ refers to ``true'', and is not a parameter.

\subsection{First Round of Conditioning}\label{subsec:truefirstcond}
\begin{definition}[Variables]
For $j\in [n]$, define the random variable
$$R_{j}^{\true}=\sum_{\ell\neq i}p_{j\ell}(1-G_{i\ell}\vee H_{i\ell}).$$
Also define the random variable
$$R_{0}^{\true}=\sum_{j\neq i}G_{ij}\vee H_{ij}.$$
\end{definition}
\begin{definition}[Events]
Let $\Gamma^{\true}$ be the event
$$\Gamma^{\true}:=\bigcap_{j=1}^{n}\left\{R_{j}^{\true}\geq \frac{\alpha}{2}d_{j}\right\}\cap \left\{\frac{1}{2}d_{i}\leq R_{0}^{\true}\leq 2d_{i}\right\}.$$
\end{definition}
\begin{proposition}\label{prop:trueGamma}
If $d \gtrsim \log^{2}n$ and $\delta \lesssim \log^{-2}n$, then $\PP(\Gamma^{\true})\geq 1-\exp(-\Omega(\log^{2}n))$.
\end{proposition}
\begin{proof}
By definition,
$$\E[R_{j}^{\true}]\geq \sum_{\ell\neq i}p_{j\ell}\cdot (\alpha-(1-\alpha)\delta)\geq (d_{j}-1)(\alpha-\delta).$$ 
By Hoeffding's inequality (\Cref{thm:Hoeffding}), using $\delta<\alpha /4$ and $p_{j\ell} \leq 1$ for all $j, \ell$, we have

$$\PP\left[R_{j}^{\true}\leq \frac{\alpha}{2}d_{j}\right]\leq 2\exp\left(-\frac{2\left((\alpha-\delta)(d_{j}-1)-\frac{1}{2}\alpha d_{j}\right)^{2}}{\sum_{\ell\neq i}p_{j\ell}^{2}}\right)\leq 2\exp\left(-\frac{\alpha^{2}}{8}(d_{j}-6)\right).$$
In addition, we have $d_{i}\leq \E[R_{0}^{\true}]\leq d_{i}(1+\delta)$ and thus by Chernoff bounds (\Cref{cor:Chernoff})
$$\PP[R_{0}^{\true}\geq 2d_{i}]\leq \exp\left(-\frac{1}{3}(1-3\delta)d_{i}\right),$$
$$\PP\left[R_{0}^{\true}\leq \frac{1}{2}d_{i}\right]\leq \exp\left(-\frac{1}{8}d_{i}\right).$$
It now follows from a union bound that $\PP(\Gamma^{\true})\geq 1-\exp(-\Omega(\log^{2}n))$.
\end{proof}
\begin{definition}[Conditioning]\label{def:tp1rd}
Let $\mathcal{F}^{\true}$ be the $\sigma$-algebra generated by the variables $G_{ij}\vee H_{ij}$ for all $j\in[n]$. 
\end{definition}
Clearly, $\Gamma^{\true}$ is $\mathcal{F}^{\true}$-measurable.

Let $J^{\true}$ be the union of neighbors of $i$ in $G$ and $H$ (see also \eqref{eq:truedefofJ} below). Note that for $j\in J^{\true}$, we can decompose the random vectors $(X_{j},Y_{j})$ into $(\sum_{\ell\notin J^{\true}}G_{j\ell},\sum_{\ell\notin J^{\true}}H_{j\ell})+(\sum_{\ell\in J^{\true}}G_{j\ell},\sum_{\ell\in J^{\true}}H_{j\ell})$, with the first part being mutually independent as $j$ varies over $J^{\true}$ while the second part not, and with the first part being independent from the second part. For $j\in J^{\true}$, the random variable $R_{j}^{\true}$ stands for the conditional expectation of the first part, and the random variable $R_{0}^{\true}$ stands for the size of $J^{\true}$. In brief, the event $\Gamma^{\true}$ says that the independent parts contain sufficient information (where the assumption that $p_{ij}$ is uniformly bounded by $1-\alpha$ is important), and at the same time the dependence can be controlled. This will facilitate our analysis after removing some dependence by conditioning.

\begin{remark}\label{remark:tp1rd}
Note that we do not condition on the complete information about the neighborhoods of $i$ in the two graphs, and this is different from the conditioning used in \Cref{subsec:fakefirstcond} or in the proof of \cite[Lemma 1]{DMWX21}. The reason for this lies in the symmetric difference of neighborhoods, i.e. $\{j\in[n]:G_{ij}\neq H_{ij}\}$. In the true pair analysis of \cite{DMWX21}, the contribution to the distance $D(i,i)$ from those vertices $j$ with $G_{ij}=1, H_{ij}=0$ and the contribution from those vertices $j$ with $G_{ij}=0, H_{ij}=1$ are shown to be \textit{separately} controllable. However, in our context, the contribution from each group of vertices can be much larger than acceptable. So instead of analyzing the two groups separately, we ``mix'' them together in a single set $J^{\true}$ (see \eqref{eq:truedefofJ}) and seek to show that the contributions from the two groups effectively \textit{cancel} with each other. By not revealing which vertices fall into which groups in $\mathcal{F}^{\true}$, we hope to leave room for a \textit{probabilistic} analysis of the cancellation. See \eqref{eq:neighborhoodsymmetry} for a particularly important place where this advantage is harnessed.
\end{remark}

\subsection{Second Round of Conditioning}\label{subsec:trueseccond}

\begin{definition}[Sets]
Define the set
\begin{equation}\label{eq:truedefofJ}
J^{\true}=\{j\in[n]\setminus\{i\}:G_{ij}\vee H_{ij}=1\}.
\end{equation}
\end{definition}
Note that conditional on $\mathcal{F}^{\true}$, it is a deterministic set. 

Interestingly, the following random variable and its associated event will appear several times in our subsequent analysis (e.g.~\Cref{lem:trueconcentration,lem:T1bound,lem:trueexpectation}). For some intuition behind the definition, each term $\frac{1}{d_{j}}+\frac{1}{d_{k}}$ corresponds to the difference an edge can make to the square roots of the degrees of its endpoints. See for example Claim 2 in the proof of \Cref{lem:trueconcentration}.

\begin{definition}[Variables]
Define the random variable
$$S^{\true}=\sum_{\{j,\ell\}\subset J^{\true}}\left(\frac{1}{d_{j}}+\frac{1}{d_{\ell}}\right)\mathbbm{1}\{G_{j\ell} \neq H_{j\ell}\}.$$
\end{definition}
\begin{definition}[Events]
Let $\Pi^{\true}$ be the event
$$\Pi^{\true}:=\{S^{\true}\leq 10\delta d_{i}\}.$$
\end{definition}
\begin{proposition}\label{prop:trueTheta}
If $d \gtrsim\log^{2}n $ and $\delta \gtrsim \log^{-2}n$, then on $\Gamma^{\true}$ we have $\PP(\Pi^{\true}|\mathcal{F}^{\true})\geq 1-\exp(-\Omega(\log^{2}n))$.
\end{proposition}
\begin{proof}
We have on $\Gamma^{\true}$ 
$$\E[S^{\true}|\mathcal{F}^{\true}]=\sum_{\{j,\ell\}\subset J^{\true}}2\delta_{j\ell} p_{j\ell}\left(\frac{1}{d_{j}}+\frac{1}{d_{\ell}}\right) \leq 2\delta|J^{\true}|\leq 4\delta d_{i}.$$ 
Note that the conditional distribution of $S^{\true}$ given $\mathcal{F}^{\true}$ is a sum of independent variables, all of which have ranges at most $2/d$, so by Chernoff bounds as in \Cref{cor:Chernoff} (note that here we need to scale the range by $2/d$),
\[\PP[S^{\true}\geq 10\delta d_{i}|\mathcal{F}^{\true}]\leq \exp\left(-2\cdot\frac{d}{2}\delta d_{i}\right).\qedhere\]
\end{proof}
\begin{definition}[Conditioning]\label{def:tp2rd}
Let $\mathcal{G}^{\true}$ be the $\sigma$-algebra generated by $\mathcal{F}^{\true}$ and the variables $\mathbbm{1}\{G_{j\ell} \neq H_{j\ell}\}$ and $ G_{j\ell}\vee H_{j\ell}$ for $j,\ell\in J^{\true}$.
\end{definition}
Clearly $\Pi^{\true}$ is $\mathcal{G}^{\true}$-measurable. For the remaining of \Cref{sec:truepair}, we always condition with respect to $\mathcal{G}^{\true}$ and assume that $\Gamma^{\true}\cap\Pi^{\true}$ holds.

\begin{remark}\label{rmk:tp2rd}
The choice of information contained in $\mathcal{G}^{\true}$ is the result of a subtle balancing act. In case $G_{j\ell}=H_{j\ell}$, a revelation of the values of $\mathbbm{1}\{G_{j\ell} \neq H_{j\ell}\}$ and $ G_{j\ell}\vee H_{j\ell}$ also reveals the complete information regarding $G_{j\ell}$ and $H_{j\ell}$ (this fact will be used in \Cref{prop:nilT2}). On the other hand, if $G_{j\ell}\neq H_{j\ell}$, then after conditioning on $\mathcal{G}^{\true}$ we know that $G_{j\ell}$ and $H_{j\ell}$ have different values, but nothing more; it remains uniformly random as to which variable takes value 1. For the purpose of bounding the expectation of the distance $D(i,k)$ (throughout \Cref{subsec:trueexpectation,subsec:upperbonT2}), it is crucial that we do not lose the symmetry of the random vector $(G_{j\ell},H_{j\ell})$.
\end{remark}

\subsection{Interlude: Truncation}\label{subsec:truncation}

In this subsection, we introduce a truncation operation on the variables $X_{j}$ and $Y_{j}$, for $j\in [n]$. While the purpose is again to exclude exceptional events, in order not to ruin the independence among the vectors $(G_{j\ell},H_{j\ell})$, we use a truncation rather than a conditioning here.

Take a fixed $j\in [n]$, and define the truncated variable $X_{j}'$ to be equal to
\begin{equation}\label{eq:defofXprime}
\begin{cases}
\E[X_{j}| \mathcal{G}^{\true}]-\log n\sqrt{\E[X_{j}| \mathcal{G}^{\true}]}& \text{if }X_j \leq \E[X_{j}| \mathcal{G}^{\true}]-\log n\sqrt{\E[X_{j}| \mathcal{G}^{\true}]},\\
X_j & \text{if }\E[X_{j}| \mathcal{G}^{\true}]-\log n\sqrt{\E[X_{j}| \mathcal{G}^{\true}]}<X_{j}<\E[X_{j}| \mathcal{G}^{\true}]+\log n\sqrt{\E[X_{j}| \mathcal{G}^{\true}]},\\
\E[X_{j}| \mathcal{G}^{\true}]+\log n\sqrt{\E[X_{j}| \mathcal{G}^{\true}]} & \text{if }\E[X_{j}| \mathcal{G}^{\true}]+\log n\sqrt{\E[X_{j}| \mathcal{G}^{\true}]}\leq X_j.
\end{cases}
\end{equation}
Similarly define variables $Y_{j}'$ for $j\in [n]$. Define accordingly

\begin{equation}\label{eq:defofAprime}
A_{ij}'(t)=\mathbbm{1}\left\{\sqrt{X_{j}'}\in I_{t}\right\}\cdot\mathbbm{1}\{G_{ij}=1\},
\end{equation}
\begin{equation}\label{eq:defofBprime}
B_{ij}'(t)=\mathbbm{1}\left\{\sqrt{Y_{j}'}\in I_{t}\right\}\cdot\mathbbm{1}\{H_{ij}=1\},
\end{equation}
and finally define
\begin{equation}\label{eq:defofDprime}
D'(i,i)=\int_{0}^{1}\left|\sum_{j=1}^{n}A_{ij}'(t)-\sum_{j=1}^{n}B_{ij}'(t)\right|\d t.
\end{equation}
\begin{lemma}\label{lem:DequalsDprime}
If $d \gtrsim \log^{2}n$, then on $\Gamma^{\true}\cap\Pi^{\true}$ we have
$$\PP[D'(i,i)\neq D(i,i)|\mathcal{G}^{\true}]=\exp(-\Omega(\log ^{2}n)).$$
\end{lemma}
\begin{proof}
We first fix some $j\in [n]$. Since $\Gamma^{\true}$ holds, we have $\E[X_{j}|\mathcal{G}^{\true}]\geq R_{j}^{\true}\geq  \alpha d_{i}/2 \gtrsim \log^{2} n$. Denoting the value of $\E[X_{j}|\mathcal{G}^{\true}]$ by $w$, by Chernoff bounds (\Cref{cor:Chernoff}) we have

$$\PP\left[X_{j}\leq w-\log n \sqrt{w}\middle|\mathcal{G}^{\true}\right]\leq \exp\left(-\frac{1}{2}\log n\sqrt{w} \cdot \frac{\log n\sqrt{w}}{w}\right)=\exp\left(-\frac{1}{2}\log^{2}n\right)$$
and
$$\PP\left[X_{j}\geq w+\log n \sqrt{w}\middle|\mathcal{G}^{\true}\right]\leq \exp\left(-\frac{1}{3}\log n\sqrt{w}\cdot\min\left\{\frac{\log n\sqrt{w}}{w},1\right\}\right)=\exp\left(-\Omega(\log ^{2}n)\right).$$
So 
$$\PP[X_{j}\neq X_{j}'|\mathcal{G}^{\true}]\leq \PP\left[X_{j}\leq w-\log n \sqrt{w}\middle|\mathcal{G}^{\true}\right]+\PP\left[X_{j}\geq w+\log n \sqrt{w}\middle|\mathcal{G}^{\true}\right]= \exp(-\Omega(\log ^{2}n)).$$
Since we have 
$$\{D'(i,i)\neq D(i,i)\} \subset \bigcup_{j\in[n]}\{X_{j}\neq X_{j}'\},$$
the conclusion follows from a union bound on $j\in[n]$.
\end{proof}

In light of \Cref{lem:DequalsDprime}, as far as \Cref{lem:truepair} is concerned, it is equivalent to analyze the random variable $D'(i,i)$ instead of $D(i,i)$. This is what we do in the remaining of \Cref{sec:truepair}.

\subsection{Concentration of Distance}\label{subsec:trueconcentration}
The analysis here is slightly more complicated than in \Cref{subsec:fakeconcentrate}. Here, even after the two rounds of conditioning, the random vectors $(X_{j},Y_{j})$ (for $j\in J^{\true}$) are not mutually independent. This is because randomness might still exist in the edges $(G_{j\ell},H_{j\ell})$ (for $\{j,\ell\}\subset J^{\true}$) due to our choice of $\mathcal{G}^{\true}$ (see \Cref{rmk:tp2rd} for more explanation), and such randomness leads to correlation between the vectors $(X_{j},Y_{j})$ and $(X_{\ell},Y_{\ell})$. So this time, instead of regarding $D'(i,i)$ as a function of the vectors $(X_{j},Y_{j})$, we regard $D'(i,i)$ as a function directly of the independent Bernoulli vectors $(G_{j\ell},H_{j\ell})$ (though some of them will be bundled together and analyzed as a whole). Let us begin with the following lemma.

\begin{lemma}\label{lem:diffofroot}
Assume that $(G,H)$ follows an  $(\alpha,4\alpha^{-1}\log^{2}n,\delta)$-edge-correlated distribution. Fix $\omega\sim (\mathcal{G}^{\true},\Gamma^{\true}\cap\Pi^{\true})$ and a pair $\{j,\ell\}\subset J^{\true}$. If we fix all other entries in $G,H$ and flip the value of $G_{j\ell}$, then the value of $\sqrt{X_{j}'}$ (as defined in \eqref{eq:defofXprime}) is changed by at most $\sqrt{\frac{2}{\alpha d_{j}}}$. 
\end{lemma}
\begin{proof}
Let $w$ be the value of $\E[X_{j}|\mathcal{G}^{\true}]$ at the realization $\omega$. 
Let $a$ be the value of $X_{j}'$ before flipping and $b$ be the value after flipping. Then $|a-b|\leq 1$ and $a,b\geq w-\log n\sqrt{w}$. Since $\Gamma^{\true}$ holds, we have $w \geq R_{j}^{\true} \geq \alpha d_{j}/2\geq 2\log^{2} n$, using the assumption $d_{j}\geq 4\alpha^{-1}\log^{2}n$. So
\[\left|\sqrt{a}-\sqrt{b}\right|\leq \frac{1}{\sqrt{a}+\sqrt{b}}\leq \frac{1}{2\sqrt{w-\log n \sqrt{w}}}\leq \frac{1}{\sqrt{w}}\leq \sqrt{\frac{2}{\alpha d_{j}}}.\qedhere\]
\end{proof}

\begin{lemma}[Concentration]\label{lem:trueconcentration}
Assume that $(G,H)$ follows an  $(\alpha,4\alpha^{-1}\log^{2}n,\delta)$-edge-correlated distribution. On $\Gamma^{\true}\cap\Pi^{\true}$, we have
$$\PP\left[D'(i,i)\geq \E[D'(i,i)|\mathcal{G}^{\true}]+\frac{50\sqrt{d_{i}\log n}}{L}+50\sqrt{\alpha^{-1}d_{i}\delta\log n}\middle|\mathcal{G}^{\true}\right]\leq O(n^{-3}).$$
\end{lemma}
\begin{proof}
In this proof, all randomness is considered under the conditioning of $\mathcal{G}^{\true}$. 

Note that the Bernoulli vectors $(G_{j\ell},H_{j\ell})$, for $\{j,\ell\}\subset [n]$, are still independent of each other after conditioning on $\mathcal{G}^{\true}$. We group the vectors into certain bundles as follows. 

\begin{itemize}
\setlength\itemsep{0pt}
\item For $\{j,\ell\}$ such that $|\{j,\ell\}\cap J^{\true}|=0$, the vector $(G_{j\ell},H_{j\ell})$ has no influence on $D'(i,i)$.

\item For $\{j,\ell\}$ such that $|\{j,\ell\}\cap J^{\true}|=1$, we group the vectors $(G_{j\ell},H_{j\ell})$ according to the intersection $\{j,\ell\}\cap J^{\true}$. Formally, for $j\in J^{\true}$,
$$F_{j}:=\left((G_{j\ell},H_{j\ell})\right)_{\ell\not\in J^{\true}}.$$
Note that $(G_{ij},H_{ij})$ is also included in the bundle $F_{j}$.

\item For $\{j,\ell\}$ such that $|\{j,\ell\}\cap J^{\true}|=2$, we assign $(G_{j\ell},H_{j\ell})$ individually to a bundle $F_{j\ell}$. Note that by the second round of conditioning, $(G_{j\ell},H_{j\ell})$ is already fixed unless $\mathbbm{1}\{G_{j\ell} \neq H_{j\ell}\}$ is known to be 1.  
\end{itemize}
Now, we consider $D'(i,i)$ as a function of these bundles, and upper-bound the marginal difference each bundle can cause on the value of $D'(i,i)$ when the values of all other bundles are fixed, for the purpose of applying McDiarmid's inequality. In the remaining part of the proof, we use superscript $\mathsf{new}$ to denote the functions and values after a change in the bundle $F_{j}$ or $F_{j\ell}$.

\begin{claim}
For $j\in J^{\true}$, the marginal difference caused by $F_{j}$ (viewed as a whole) is no more than $8/L$.
\end{claim}
\begin{proof}[Proof of Claim 1]\renewcommand{\qedsymbol}{$\blacksquare$}
As defined in \eqref{eq:defofDprime}, the value of $D'(i,i)$ is determined by the $2n$ functions $A_{ir}'(\cdot)$ and $B_{ir}'(\cdot)$, for $r\in [n]$. Keeping all other bundles fixed, a change in the bundle $F_{j}$ will only change two of the $2n$ functions: $A_{ij}'(t)$ and $B_{ij}'(t)$. If $G_{ij}=0$, then $A_{ij}'(\cdot)$ and $A_{ij}'(\cdot)^{\mathsf{new}}$ are both zero functions (using the superscript $\mathsf{new}$ to denote the values after a change in the bundle $F_{j}$). Otherwise, in the same way as \Cref{prop:structureofA}, $A_{ij}'(\cdot)$ and $A_{ij}'(\cdot)^{\mathsf{new}}$ are both indicator functions of a cyclic interval of $2/L$. Since similar results also hold for $B_{ij}'(\cdot)$ and $B_{ij}'(\cdot)^{\mathsf{new}}$, we have
\begin{align*}
\left|D'(i,i)^{\mathsf{new}}-D'(i,i)\right|&\leq \sum_{r=1}^{n}\int_{0}^{1}\Big(\left|A_{ir}'(t)^{\mathsf{new}}-A_{ir}'(t)\right|+\left|B_{ir}'(t)^{\mathsf{new}}-B_{ir}'(t)\right|\Big)\d t\\
&\leq\int_{0}^{1}\left|A_{ij}'(t)^{\mathsf{new}}\right|\d t+\int_{0}^{1}\left|A_{ij}'(t)\right|\d t+\int_{0}^{1}\left|B_{ij}'(t)^{\mathsf{new}}\right|\d t+\int_{0}^{1}\left|B_{ij}'(t)\right|\d t\\
&\leq \frac{2}{L}+\frac{2}{L}+\frac{2}{L}+\frac{2}{L}=\frac{8}{L}.\qedhere
\end{align*}
\end{proof}
\begin{claim}
For $\{j,\ell\}\subset J^{\true}$ such that $\mathbbm{1}\{G_{j\ell} \neq H_{j\ell}\}$ is known to be 1, the marginal difference caused by $F_{j\ell}$ is no more than
$$\sqrt{\frac{64}{\alpha}\left(\frac{1}{d_{j}}+\frac{1}{d_{\ell}}\right)}.$$
\end{claim}
\begin{proof}[Proof of Claim 2]\renewcommand{\qedsymbol}{$\blacksquare$}
As defined in \eqref{eq:defofDprime}, the value of $D'(i,i)$ is determined by the $2n$ functions $A_{ir}'(\cdot)$ and $B_{ir}'(\cdot)$, for $r\in [n]$. Keeping all other bundles fixed, a change in the bundle $F_{j\ell}$ will only change four of the $2n$ functions: $A_{ij}'(t),A_{i\ell}'(t),B_{ij}'(t)$ and $B_{i\ell}'(t)$. If $G_{ij}=0$, then $A_{ij}'(\cdot)$ and $A_{ij}'(\cdot)^{\mathsf{new}}$ are both zero functions (using the superscript $\mathsf{new}$ to denote the values after a change in the bundle $F_{j\ell}$). Otherwise, in the same way as \Cref{prop:structureofA}, $A_{ij}'(\cdot)$ and $A_{ij}'(\cdot)^{\mathsf{new}}$ are both indicator functions of a cyclic interval of $2/L$,  centered at $\sqrt{X_{j}'}$ and $\sqrt{X_{j}'^{\,\mathsf{new}}}$ respectively. From \Cref{lem:diffofroot} we know $$\left|\sqrt{X_{j}'^{\,\mathsf{new}}}-\sqrt{X_{j}'}\right|\leq \sqrt{\frac{2}{\alpha d_{j}}}.$$ 
So the symmetric difference of the two cyclic intervals, i.e.
$$\left\{t\in [0,1):A_{ij}'(t)\neq A_{ij}'(t)^{\mathsf{new}}
\right\}$$ has Lebesgue measure no more than $2\sqrt{\frac{2}{\alpha d_{j}}}$. Since similar results also hold for other three pairs of functions, we have
\begin{align*}
\left|D'(i,i)^{\mathsf{new}}-D'(i,i)\right|&\leq  \sum_{r=1}^{n}\int_{0}^{1}\Big(\left|A_{ir}'(t)^{\mathsf{new}}-A_{ir}'(t)\right|+\left|B_{ir}'(t)^{\mathsf{new}}-B_{ir}'(t)\right|\Big)\d t\\
&= \int_{0}^{1}\left|A_{ij}'(t)^{\mathsf{new}}-A_{ij}(t)\right|\d t+\int_{0}^{1}\left|B_{ij}'(t)^{\mathsf{new}}-B_{ij}(t)\right|\d t+ \\
&\quad\int_{0}^{1}\left|A_{i\ell}'(t)^{\mathsf{new}}-A_{i\ell}(t)\right|\d t + \int_{0}^{1}\left|B_{i\ell}'(t)^{\mathsf{new}}-B_{i\ell}(t)\right|\d t\\
&\leq  2\sqrt{\frac{2}{\alpha d_{j}}}+2\sqrt{\frac{2}{\alpha d_{j}}}+2\sqrt{\frac{2}{\alpha d_{\ell}}}+2\sqrt{\frac{2}{\alpha d_{\ell}}}\\
&\leq \sqrt{\frac{64}{\alpha}\left(\frac{1}{d_{j}}+\frac{1}{d_{\ell}}\right)},
\end{align*}
where we used Cauchy-Schwarz in the last inequality.
\end{proof}

At any $\omega\sim(\mathcal{G}^{\true},\Gamma^{\true}\cap\Pi^{\true})$, combining the above two claims, we can conclude by McDiarmid's inequality (\Cref{thm:Mcdiarmid}) that with conditional probability (on $\mathcal{G}^{\true}$) at least $1-O\left(n^{-3}\right)$,
\begin{align*}
D'(i,i)-\E[D'(i,i)|\mathcal{G}^{\true}]&\leq 
\sqrt{\frac{1}{2}\cdot 3\log n\cdot \left(\sum_{j\in J^{\true}}\left(\frac{8}{L}\right)^{2}+\sum_{\{j,\ell\}\subset J^{\true}}\mathbbm{1}\{G_{j\ell}\neq H_{j\ell}\}\frac{64}{\alpha}\left(\frac{1}{d_{j}}+\frac{1}{d_{\ell}}\right)\right)}\\
&=\sqrt{\frac{1}{2}\cdot 3\log n\cdot\left(|J^{\true}|\cdot\left(\frac{8}{L}\right)^{2}+\frac{64}{\alpha}S^{\true}\right)}\\
&\leq \frac{50\sqrt{d_{i}\log n}}{L}+50\sqrt{\alpha^{-1}d_{i}\delta\log n}.
\end{align*}
In the last inequality we used $|J^{\true}|=R_{0}^{\true}\leq 2d_{i}$ and $S^{\true}\leq 10\delta d_{i}$, as guaranteed by $\Gamma^{\true}\cap\Pi^{\true}$.
\end{proof}

\subsection{Expectation of Distance}\label{subsec:trueexpectation}
Now we finally arrive at the crux of the proof: we need to give an upper bound on $\E[D'(i,i)|\mathcal{G}^{\true}]$. Let us first define two families of quantities.
\begin{definition}
Define for $j\in J^{\true}$
$$T_{1}(j):=\int_{0}^{1}\PP\left[A_{ij}'(t)\neq B_{ij}'(t)\middle|\mathcal{G}^{\true}\right]\d t$$
and for $\{j,\ell\}\subset J^{\true}$
\begin{equation}\label{eq:defofT2}
T_{2}(j,\ell)=\int_{0}^{1}\Big(\PP\left[A_{ij}'(t)=A_{i\ell}'(t)=1\middle|\mathcal{G}^{\true}\right]-\PP\left[A_{ij}'(t)=B_{i\ell}'(t)=1\middle|\mathcal{G}^{\true}\right]\Big)\d t.
\end{equation}
\end{definition}

The following lemma then reduces the task of upper-bounding $\E\left[D'(i,i)\middle|\mathcal{G}^{\true}\right]$ to dealing with $T_{1}$ and $T_{2}$ separately. 

\begin{lemma}\label{lem:turntrueexptoT}
$\E\left[D'(i,i)^{2}\middle|\mathcal{G}^{\true}\right]\leq \sum_{j\in J^{\true}}T_{1}(j)+4\sum_{\{j,\ell\}\subset J^{\true}}T_{2}(j,\ell)$.
\end{lemma}
\begin{proof}
We have by Cauchy-Schwarz inequality
\begin{equation}\label{eq:D'CauchySchwarz}
\E\left[D'(i,i)^{2}\middle|\mathcal{G}^{\true}\right]\leq \int_{0}^{1}\E\left[\left|\sum_{j=1}^{n}(A_{ij}'(t)-B_{ij}'(t))\right|^{2} \mid \mathcal{G}^{\true} \right]\d t.
\end{equation}
For each $t\in[0,1)$, the integrand on the right hand side of \eqref{eq:D'CauchySchwarz} expands to
\begin{align}
&\quad\E\left[\left|\sum_{j\in J^{\true}}(A_{ij}'(t)-B_{ij}'(t))\right|^{2}\mid \mathcal{G}^{\true}\right]\nonumber\\
&=\sum_{j\in J^{\true}}\E\left[(A_{ij}'(t)- B_{ij}'(t))^{2}\middle|\mathcal{G}^{\true} \right]+2\sum_{\{j,\ell\}\subset J^{\true}}\E\left[(A_{ij}'(t)-B_{ij}'(t))(A_{i\ell}'(t)-B_{i\ell}'(t))\middle|\mathcal{G}^{\true}\right].\label{eq:D'expansion}
\end{align}
For each $\{j,\ell\}\subset J^{\true}$, the second term on the right hand side of \eqref{eq:D'expansion} expands to
\begin{align}
&\quad\E\left[(A_{ij}'(t)-B_{ij}'(t))(A_{i\ell}'(t)-B_{i\ell}'(t))\middle|\mathcal{G}^{\true}\right]\nonumber\\
&=\left(\PP[A_{ij}'(t)=A_{i\ell}'(t)=1|\mathcal{G}^{\true}]+\PP[B_{ij}'(t)=B_{i\ell}'(t)=1|\mathcal{G}^{\true}]\right)-\nonumber\\&\quad \left(\PP[A_{ij}'(t)=B_{i\ell}'(t)=1|\mathcal{G}^{\true}]+\PP[B_{ij}'(t)=A_{i\ell}'(t)=1|\mathcal{G}^{\true}]\right) \nonumber\\
&=2\PP[A_{ij}'(t)=A_{i\ell}'(t)=1 |\mathcal{G}^{\true}]-2\PP[A_{ij}'(t)=B_{i\ell}'(t)=1 |\mathcal{G}^{\true}].\label{eq:D'expansion2}
\end{align}
In the last equality, we used the fact that $(A_{ij}'(t),A_{i\ell}'(t))$ has the same conditional distribution as $(B_{ij}'(t),B_{i\ell}'(t))$, and $(A_{ij}'(t),B_{i\ell}'(t))$ has the same conditional distribution as $(B_{ij}'(t),A_{i\ell}'(t))$.
The conclusion then follows by combining \eqref{eq:D'CauchySchwarz}, \eqref{eq:D'expansion} and \eqref{eq:D'expansion2}.
\end{proof}
We now give an upper bound on $T_{1}$ and will defer the analysis of $T_{2}$ to the next subsection. Note that bounding $T_{1}$ somewhat corresponds to the true pair analysis in \cite{DMWX21} and the analysis of $T_{2}$ is an entirely new ingredient of the proof. It is interesting that due to our continuous version of the balls-into-bins distance, the analysis of $T_{1}$ bypasses the rather complicated Lemma 9 of \cite{DMWX21}.
\begin{lemma}\label{lem:T1bound}
Assume that $(G,H)$ follows an  $(\alpha,4\alpha^{-1}\log^{2}n,\delta)$-edge-correlated distribution. On $\Gamma^{\true}\cap\Pi^{\true}$, we have
$$\sum_{j\in J^{\true}}T_{1}(j)\leq 50d_{i}\left(\frac{\delta}{L}+\sqrt{\alpha^{-1}\delta}\right).$$
\end{lemma}
\begin{proof}
Let us first bound $T_1(j)$ for a given $j\in J^{\true}$. Let $\mu(\cdot)$ denote the Lebesgue measure on $\mathbb{R}$. By Fubini's theorem we have 
\begin{equation}\label{eq:T1fubini}
\int_{0}^{1}\PP\left[A_{ij}'(t)\neq B_{ij}'(t)\middle|\mathcal{G}^{\true}\right]\d t =\E\Big[\mu\left(\{t\in [0,1):A_{ij}'(t)\neq B_{ij}'(t)\}\right)\Big|\mathcal{G}^{\true}\Big].
\end{equation}
We again use analogues of \Cref{prop:structureofA} for $A_{ij}'(\cdot)$ and $B_{ij}'(\cdot)$. The functions $A_{ij}'(\cdot)$ and $B_{ij}'(\cdot)$ are both indicator functions with supports having Lebesgue measures no more than $2/L$, so the symmetric difference between the supports of $A_{ij}'(\cdot)$ and $B_{ij}'(\cdot)$, that is, 
$$\left\{t\in [0,1): A_{ij}'(t)\neq B_{ij}'(t)\right\},$$
has Lebesgue measure no more than $4/L$.

In the case $G_{ij}=H_{ij}$, there is another way to upper-bound the measure of this symmetric difference. In fact, since $j\in J^{\true}$, the event $\{G_{ij}=H_{ij}\}$ implies $G_{ij}=H_{ij}=1$. In this case, $A_{ij}'(t)$ (resp. $B_{ij}'(\cdot)$) is an indicator function of a cyclic interval of length $2/L$ with center $\sqrt{X_{j}'}$ (resp. $\sqrt{Y_{j}'}$). So the symmetric difference of the two cyclic intervals has Lebesgue measure no more than  2$\left|\sqrt{X_{j}'}-\sqrt{Y_{j}'}\right|$. 
Therefore the right hand side of $\eqref{eq:T1fubini}$ is upper-bounded by 
$$\PP[G_{ij}\neq H_{ij}|\mathcal{G}^{\true}]\cdot \frac{4}{L}+\PP[G_{ij}=H_{ij}|\mathcal{G}^{\true}]\cdot\E\left[2\left|\sqrt{X_{j}'}-\sqrt{Y_{j}'}\right|\mid G_{ij}=H_{ij},\mathcal{G}^{\true}\right].$$
Since $\PP[G_{ij} \neq H_{ij} | \mathcal{G}^{\true}]=2\delta_{ij}/(1+\delta_{ij}) \leq 2\delta,$
it is further upper-bounded by 
$$\frac{8\delta}{L}+\E\left[2\left|\sqrt{X_{j}'}-\sqrt{Y_{j}'}\right|\mid G_{ij}=H_{ij},\mathcal{G}^{\true}\right].$$
Recall that by $\Gamma^{\true}$ and the assumption $d_{j}\geq 4\alpha^{-1}\log^{2}n$, we have $\E[X_j|\mathcal{G}^{\true}] \geq \alpha d_j/2 \geq 2\log^{2}n$. So $$X_{j}' \geq \E[X_j|\mathcal{G}^{\true}]-\log n\sqrt{\E[X_j|\mathcal{G}^{\true}]} \geq \left(1-\frac{1}{\sqrt{2}}\right)\E[X_j|\mathcal{G}^{\true}] \geq \frac{\alpha d_j}{8}.$$ 
Similarly, $Y_{j}' \geq \alpha d_j/8$. It follows that  
\begin{align}
T_1(j)&\leq \frac{8\delta}{L}+\E\left[2\left|\sqrt{X_{j}'}-\sqrt{Y_{j}'}\right|\mid G_{ij}=H_{ij},\mathcal{G}^{\true}\right]\nonumber\\
&\leq \frac{8\delta}{L}+2\sqrt{\frac{2}{\alpha d_{j}}}\E\left[|X_{j}-Y_{j}|\middle|G_{ij}=H_{ij},\mathcal{G}^{\true}\right] \nonumber\\
&\leq \frac{8\delta}{L}+2\sqrt{\frac{2}{\alpha d_{j}}}\E\left[\left|\sum_{\ell\in J^{\true}}(G_{j\ell}-H_{j\ell})\right| \mid \mathcal{G}^{\true}\right]+2\sqrt{\frac{2}{\alpha d_{j}}}\E\left[\left|\sum_{\ell\not\in J^{\true}\cup\{i\}}(G_{j\ell}-H_{j\ell}) \right|\mid \mathcal{G}^{\true}\right].\label{eq:T1RHS}
\end{align}

For the second term on the right hand side of \eqref{eq:T1RHS}, using Cauchy-Schwarz and the bound $|J^{\true}| \leq 2 d_{i}$ due to $\Gamma^{\true}$, then using the conditional independence between vectors $(G_{j\ell},H_{j\ell})$ (for $\ell \in J^{\true}$) and that the conditional expectation of $G_{j\ell}$ and $H_{j\ell}$ are the same, and finally using the bound $S^{\true}\leq 10\delta d_{i}$ due to $\Pi^{\true}$, we get the following string of three inequalities accordingly:
\begin{align}
\sum_{j\in J^{\true}}\sqrt{\frac{1}{d_{j}}}\E\left[\left|\sum_{\ell\in J^{\true}}(G_{j\ell}-H_{j\ell})\right| \mid \mathcal{G}^{\true} \right]
&\leq \sqrt{2d_{i}}\left(\sum_{j\in J^{\true}}\frac{1}{d_{j}}\E\left[\left|\sum_{\ell\in J^{\true}}(G_{j\ell}-H_{j\ell})\right|^{2} \mid \mathcal{G}^{\true}\right]\right)^{1/2}\nonumber\\
&\leq \sqrt{2d_{i}}\left(\sum_{\{j,\ell\}\subset J^{\true}}\left(\frac{1}{d_{j}}+\frac{1}{d_{\ell}}\right)\mathbbm{1}\{G_{j\ell} \neq H_{j\ell}\}\right)^{1/2}\nonumber\\
&\leq d_{i}\sqrt{20\delta}.\label{eq:2ndterm}
\end{align}

For the third term on the right hand side of \eqref{eq:T1RHS}, again using Cauchy-Schwarz and the bound $|J^{\true}| \leq 2 d_{i}$, then using the conditional independence between vectors $(G_{j\ell},H_{j\ell})$ (for $\ell\not\in J^{\true}\cup \{i\}$) and that the conditional expectation of $G_{j\ell}$ and $H_{j\ell}$ are the same, and finally using the bound $|J^{\true}| \leq 2 d_{i}$ again, we have the following series of transitions:
\begin{align}\label{eq:3rdterm}
\sum_{j\in J^{\true}}\sqrt{\frac{1}{d_{j}}}\E\left[\left|\sum_{\ell\not\in J^{\true}\cup\{i\}}(G_{j\ell}-H_{j\ell})\right|\mid \mathcal{G}^{\true}\right]&\leq  \sqrt{2d_{i}} \left(\sum_{j\in J^{\true}}\frac{1}{d_{j}}\E\left[\left|\sum_{\ell\not\in J^{\true}\cup\{i\}}(G_{j\ell}-H_{j\ell})\right|^{2}\mid \mathcal{G}^{\true}\right]\right)^{1/2} \nonumber \\
&=\sqrt{2d_{i}}\left(\sum_{j\in J^{\true}}\frac{1}{d_{j}}\sum_{\ell\not\in J^{\true}\cup \{i\}}2\delta_{j\ell} p_{j\ell}\right)^{1/2} \leq  2d_{i}\sqrt{2\delta}.
\end{align}
Summing over $j\in J^{\true}$ on both sides of \eqref{eq:T1RHS}, then applying \eqref{eq:2ndterm}, \eqref{eq:3rdterm} and the bound $|J^{\true}| \leq 2 d_{i}$, we have
\begin{align*}
\sum_{j\in J^{\true}}T_{1}(j)&\leq |J^{\true}|\cdot \frac{8\delta}{L}+d_{i}\sqrt{160\alpha^{-1}\delta}+8d_{i}\sqrt{\alpha^{-1}\delta}\\
&\leq 50d_{i}\left(\frac{\delta}{L}+\sqrt{\alpha^{-1}\delta}\right).\qedhere
\end{align*}
\end{proof}

\subsection{Upper bound on $T_{2}$}\label{subsec:upperbonT2}

In this subsection, we provide an upper bound for $T_{2}$. As can be seen in the definition \eqref{eq:defofT2}, we need to analyze the joint distributions of $(A_{ij}'(t),A_{i\ell}'(t))$ and of $(A_{ij}'(t),B_{i\ell}'(t))$, conditional on $\mathcal{G}^{\true}$. A particularly simple case is captured by the following proposition.
\begin{proposition}\label{prop:nilT2}
For $j,\ell\in J^{\true}$, if $G_{j\ell}=H_{j\ell}$ then $T_{2}(j,\ell)=0$.
\end{proposition}
\begin{proof}
If $G_{j\ell}=H_{j\ell}$, the values of $G_{j\ell}$ and $H_{j\ell}$ are known after conditioning on $\mathcal{G}^{\true}$ (see \Cref{rmk:tp2rd}). So after the conditioning, the random function $A_{ij}'(\cdot)$ depends only on the randomness in $\left\{G_{jr}\right\}_{r\neq \ell}$ and $A_{i\ell}'(\cdot)$ only depends on the randomness in $\left\{G_{\ell r}\right\}_{r\neq j}$. This means that the random functions $A_{ij}'(\cdot)$ and $A_{i\ell}'(\cdot)$ are (conditionally) independent. Similarly, $B_{i\ell}'(\cdot)$ only depends on the randomness in $\left\{H_{\ell r}\right\}_{r\neq j}$, and $A_{ij}'(\cdot)$ is also independent with $B_{i\ell}'(\cdot)$. Since $A_{i\ell}'(\cdot)$ is identically distributed with $B_{i\ell}'(\cdot)$ (even after conditioning on $\mathcal{G}^{\true}$), we conclude that for any $t\in [0,1)$, the (conditional) joint distribution of $(A_{ij}'(t),A_{i\ell}'(t))$ is the same as that of $(A_{ij}'(t),B_{i\ell}'(t))$. Consequently by \eqref{eq:defofT2} we must have $T_{2}(j,\ell)=0$. 
\end{proof}

For the remaining of \Cref{subsec:upperbonT2}, we fix $\{j,\ell\} \subset J^{\true}$ and a realization $\omega\sim(\mathcal{G}^{\true},\Gamma^{\true}\cap\Pi^{\true})$ such that $G_{j\ell}\neq H_{j\ell}$ at $\omega$. Our goal is to give an upper bound on the value of $T_{2}(j,\ell)$ at $\omega$.

\subsubsection{Preparations: the Alias World}\label{subsubsec:aliasworld}
Previously, we have based our analysis exclusively on the randomness existing in the random graphs $G$ and $H$. However, as the analysis goes deeper, the random variables arising from $(G,H)$ get increasingly intertwined with each other. For the purpose of simplifying notations and reducing distractions, we create some external random variables that simulate parts of the randomness in $(G,H)$. 

In this subsection, we will make a series of definitions that form what we call ``the alias world''. The random variables here are defined on a different probability space from that of $(G,H)$, but their distributions offer a reflection of the latter. All definitions here will be put into use in \Cref{subsubsec:modelreal}.

\begin{definition}[Alias variables]\label{def:aliasvar}
Let $Z_{1}$ be a random variable that admits the same distribution as the conditional distribution of $1+\sum_{r\not\in\{i,j,\ell\}}G_{jr}$, i.e. 
$$\forall x\in \mathbb{Z},\quad \PP[Z_{1}=x]=\PP\left[1+\sum\nolimits_{r\not\in\{i,j,\ell\}}G_{jr}=x\middle|\mathcal{G}^{\true}\right](\omega).$$
Let $Z_{2}$ be a random variable admitting the same distribution as the conditional distribution of $1+\sum_{r\not\in\{i,j,\ell\}}G_{\ell r}$, i.e.
$$\forall x\in \mathbb{Z},\quad \PP[Z_{2}=x]=\PP\left[1+\sum\nolimits_{r\not\in\{i,j,\ell\}}G_{\ell r}=x\middle|\mathcal{G}^{\true}\right](\omega).$$
Finally, let $Z_{0}$ be a Bernoulli random variable with $\PP[Z_{0}=0]=\PP[Z_{0}=1]=\frac{1}{2}$. Furthermore, we assume $Z_{0},Z_{1},Z_{2}$ are mutually independent.
\end{definition}

Parallel to \Cref{subsec:truncation}, we accordingly introduce a truncation function for our alias variables.

\begin{definition}
Let $w_{1}$ be the value $\E[X_{j}|\mathcal{G}^{\true}](\omega)$, and let $w_{2}$ be the value $\E[X_{\ell}|\mathcal{G}^{\true}](\omega)$. Define two functions $\mathsf{Trc}_{s}:\mathbb{R}\rightarrow\mathbb{R}$, for $s\in\{1,2\}$, by 
$$\mathsf{Trc}_{s}(x)=\begin{cases}
w_{s}-\log n\sqrt{w_{s}} &\text{if }x\leq w_{s}-\log n\sqrt{w_{s}},\\
x &\text{if }w_{s}-\log n\sqrt{w_{s}}< x<w_{s}+\log n\sqrt{w_{s}},\\
w_{s}+\log n\sqrt{w_{s}} &\text{if }w_{s}+\log n\sqrt{w_{s}}\leq x.
\end{cases}$$
\end{definition}

\begin{definition}
Define a function $g:\mathbb{R}\rightarrow [0,2/L]$ by
$$g(x):=\max\left\{0,\frac{2}{L}-\min_{m\in\mathbb{Z}}\left|x-m\right|\right\}.$$
\end{definition}

The function $g$ will be our principal tool for analyzing $T_{2}(j,\ell)$. The reason behind its definition is the following crucial lemma.
\begin{lemma}\label{lem:rationaleg}
Recall the definition of $I_{t}$ in \eqref{eq:defofItm} and \eqref{eq:defofIt}, where we assumed $L\geq 5$. Recall that $\mu(\cdot)$ denotes the Lebesgue measure on $\mathbb{R}$. For any two real numbers $a,b$, we have the identity
$$\mu\left(\{t\in[0,1):a,b\in I_{t}\}\right)=g(a-b).$$
\end{lemma}
\begin{proof}
For any real numbers $x,y$, define their ``cyclic distance'' to be
$$\mathsf{cycd}(x,y):=\min_{m\in\mathbb{Z}}|x-y-m|.$$
Recall that $\{t\in[0,1):a\in I_{t}\}$ is the cyclic interval of length $2/L$ centered at $a$, i.e.
$$\left\{t\in [0,1):\mathsf{cycd}(t,a)\leq \frac{1}{L}\right\}.$$
So the set $\{t\in [0,1):a,b\in I_{t}\}$ is the intersection of two cyclic intervals, and hence
\begin{align*}
\mu\left(\{t:a,b\in I_{t}\}\right)&=\mu\left(\left\{t:\mathsf{cycd}(t,a)\leq \frac{1}{L}\right\}\cap\left\{t:\mathsf{cycd}(t,b)\leq \frac{1}{L}\right\}\right)\\
&=\begin{cases}
0, &\text{if }\mathsf{cycd}(a,b)\geq \frac{2}{L}\\
\frac{2}{L}-\mathsf{cycd}(a,b),&\text{if }\mathsf{cycd}(a,b)< \frac{2}{L}
\end{cases}\\
&=\max\left\{0,\frac{2}{L}-\mathsf{cycd}(a,b)\right\}=g(a-b).
\end{align*}
Note that for the second equality above to hold, we need the condition that $\frac{1}{L}<\frac{1}{4}$.
\end{proof}
We then define an event on which the variables $Z_{1}$ and $Z_{2}$ take typical values. 
\begin{definition}[Events]\label{def:trueLambda}
Let $\Lambda$ be the event
\begin{align*}
\left\{\E[Z_{1}]+1-\log n\sqrt{\E[Z_{1}]-1}\leq Z_{1}\leq \E[Z_{1}]-2+\log n\sqrt{\E[Z_{1}]-1}\right\}&\cap\\
\left\{\E[Z_{2}]+1-\log n\sqrt{\E[Z_{2}]-1}\leq Z_{2}\leq \E[Z_{2}]-2+\log n\sqrt{\E[Z_{2}]-1}\right\}&
\end{align*}
\end{definition}
\begin{proposition}\label{prop:trueLambda}
If $\E[Z_{1}],\E[Z_{2}] \gtrsim \log^{2}n$, then $\PP(\Lambda)\geq 1-\exp(-\Omega(\log ^{2}n))$. 
\end{proposition}
\begin{proof}
The statement follows directly from Chernoff bounds (in a similar way to the proof of \Cref{lem:DequalsDprime}).
\end{proof}

\begin{lemma}\label{lem:ridofTrc}
On the event $\Lambda$, we have $\mathsf{Trc}_{s}(Z_{s})=Z_{s}$ and $\mathsf{Trc}_{s}(Z_{s}+1)=Z_{s}+1$, for any $s\in\{1,2\}$.
\end{lemma}
\begin{proof}
Observe that we always have
$$\sum\nolimits_{r\notin\{i,j,\ell\}}G_{jr}\leq X_{j}\leq 2+\sum\nolimits_{r\notin\{i,j,\ell\}}G_{jr}.$$
So by \Cref{def:aliasvar}, we know that
\begin{equation}\label{eq:ZandX}
\E[Z_{1}]-1\leq \E[X_{j}|\mathcal{G}^{\true}](\omega)\leq \E[Z_{1}]+1,
\end{equation}
and similarly 
$$\E[Z_{2}]-1\leq \E[X_{\ell}|\mathcal{G}^{\true}](\omega)\leq \E[Z_{2}]+1,$$
Again letting $w_{1}$ denote the value $\E[X_{j}|\mathcal{G}^{\true}](\omega)$ and $w_{2}$ the value $\E[X_{\ell}|\mathcal{G}^{\true}](\omega)$, the condition of $\Lambda$ easily implies 
$$w_{s}-\log n\sqrt{w_{s}}\leq Z_{s}< Z_{s}+1\leq w_{s}+\log n\sqrt{w_{s}},$$
and hence $\mathsf{Trc}_{s}(Z_{s})=Z_{s}$ and $\mathsf{Trc}_{s}(Z_{s}+1)=Z_{s}+1$, for $s\in \{1,2\}$.
\end{proof}

Finally, we will define a quartet of variables which will be the focus of attention in \Cref{subsubsec:aliasanal}.

\begin{definition}[Variables]\label{def:defofK}
Define four random variables
\begin{align*}
K_{00}&=\mathbbm{1}\{\Lambda\} g\left(\sqrt{Z_{1}}-\sqrt{Z_{2}}\right),\\
K_{01}&=\mathbbm{1}\{\Lambda\}g\left(\sqrt{Z_{1}}-\sqrt{Z_{2}+1}\right),\\
K_{10}&= \mathbbm{1}\{\Lambda\}g\left(\sqrt{Z_{1}+1}-\sqrt{Z_{2}}\right),\\
K_{11}&= \mathbbm{1}\{\Lambda\}g\left(\sqrt{Z_{1}+1}-\sqrt{Z_{2}+1}\right).
\end{align*}
\end{definition}
\subsubsection{Modeling the Real World}\label{subsubsec:modelreal}
We now return from the alias world to the original random graph pair $(G,H)$ and show how our alias constructions are able to simulate the actual randomness in $\mathcal{G}^{\true}$. The following two lemmas draw an intimate connection between the alias variables and the quantities of interest in \eqref{eq:defofT2}.

\begin{lemma}\label{lem:modelingAA}
Assume that $L\geq 5$. We have the identity
\begin{align*}
&\quad \int_{0}^{1}\PP[A_{ij}'(t)=A_{i\ell}'(t)=1|\mathcal{G}^{\true}](\omega)\d t\\
&=\PP[G_{ij}=G_{i\ell}=1|\mathcal{G}^{\true}](\omega)\cdot\E\left[g\left(\sqrt{\mathsf{Trc}_{1}(Z_{1}+Z_{0})}-\sqrt{\mathsf{Trc}_{2}(Z_{2}+Z_{0})}\right)\right].
\end{align*}
\end{lemma}

\begin{proof}
By the definition in \eqref{eq:defofAprime}, for any $t\in[0,1)$, we have the inclusion of events $\{A_{ij}'(t)=A_{i\ell}'(t)=1\}\subset \{G_{ij}=G_{i\ell}=1\}$. So we may write
\begin{align}
&\quad\int_{0}^{1}\PP[A_{ij}'(t)=A_{i\ell}'(t)=1|\mathcal{G}^{\true}]\d t\nonumber\\
&=\PP[G_{ij}=G_{i\ell}=1|\mathcal{G}^{\true}]\cdot\int_{0}^{1}\PP\left[\sqrt{X_{j}'},\sqrt{X_{\ell}'}\in I_{t}\middle|G_{ij}=G_{i\ell}=1, \mathcal{G}^{\true}\right]\d t.\label{eq:oneventGG1}
\end{align}
On the event $\{G_{ij}=G_{i\ell}=1\}$, we can decompose
$$(X_{j},X_{\ell})=\left(\left(1+\sum\nolimits_{r\notin\{i,j,\ell\}}G_{jr}\right)+G_{j\ell},\;\left(1+\sum\nolimits_{r\notin\{i,j,\ell\}}G_{\ell r}\right)+G_{j\ell}\right).$$
Note that conditional on $\mathcal{G}^{\true}$, the three variables $\left(1+\sum\nolimits_{r\notin\{i,j,\ell\}}G_{jr}\right)$, $\left(1+\sum\nolimits_{r\notin\{i,j,\ell\}}G_{\ell r}\right)$ and $G_{j\ell}$ and the event $\{G_{ij}=G_{i\ell}=1\}$ are mutually independent. Furthermore, since $G_{j\ell}\neq H_{j\ell}$ at $\omega$, we have (see \Cref{rmk:tp2rd})
\[\PP[G_{j\ell}=1| \mathcal{G}^{\true}](\omega)=\PP[G_{j\ell}=0|\mathcal{G}^{\true}](\omega)=\frac{1}{2},\]
meaning that the conditional distribution of $G_{j\ell}$ is the same as the distribution of $Z_{0}$. 
Therefore, by \Cref{def:aliasvar}, the joint distribution of $(X_{j},X_{\ell})$ 
conditional on both $\mathcal{G}^{\true}$ and the event $\{G_{ij}=G_{i\ell}=1\}$ is the same as the distribution of $(Z_{1}+Z_{0},Z_{2}+Z_{0})$. So
\begin{align}
&\quad\int_{0}^{1}\PP\left[\sqrt{X_{j}'},\sqrt{X_{\ell}'}\in I_{t}\middle|G_{ij}=G_{i\ell}=1, \mathcal{G}^{\true}\right](\omega)\d t\nonumber\\
&=\int_{0}^{1}\PP\left[\sqrt{\mathsf{Trc}_{1}(Z_{1}+Z_{0})},\sqrt{\mathsf{Trc}_{2}(Z_{2}+Z_{0})}\in I_{t}\right]\d t\nonumber\\
&=\E\left[\mu\left(t\in[0,1):\sqrt{\mathsf{Trc}_{1}(Z_{1}+Z_{0})},\sqrt{\mathsf{Trc}_{2}(Z_{2}+Z_{0})}\in I_{t}\right)\right]&(\text{by Fubini's theorem})\nonumber\\
&=\E\left[g\left(\sqrt{\mathsf{Trc}_{1}(Z_{1}+Z_{0})}-\sqrt{\mathsf{Trc}_{2}(Z_{2}+Z_{0})}\right)\right]&(\text{by \Cref{lem:rationaleg}}).
\label{eq:ridofGG1}
\end{align}
The conclusion follows by combining \eqref{eq:oneventGG1} and \eqref{eq:ridofGG1}.
\end{proof}

\begin{lemma}\label{lem:modelingAB}
Assume that $L\geq 5$. We have the identity
\begin{align*}
&\quad \int_{0}^{1}\PP[A_{ij}'(t)=B_{i\ell}'(t)=1|\mathcal{G}^{\true}](\omega)\d t\\
&=\PP[G_{ij}=H_{i\ell}=1|\mathcal{G}^{\true}](\omega)\cdot\E\left[g\left(\sqrt{\mathsf{Trc}_{1}(Z_{1}+Z_{0})}-\sqrt{\mathsf{Trc}_{2}(Z_{2}+1-Z_{0})}\right)\right].
\end{align*}
\end{lemma}
\begin{proof}
By definitions in \eqref{eq:defofAprime} and \eqref{eq:defofBprime}, for any $t\in[0,1)$, we have the inclusion of events $\{A_{ij}'(t)=B_{i\ell}'(t)=1\}\subset \{G_{ij}=H_{i\ell}=1\}$. So we may write
\begin{align}
&\quad\int_{0}^{1}\PP[A_{ij}'(t)=B_{i\ell}'(t)=1|\mathcal{G}^{\true}]\d t\nonumber\\
&=\PP[G_{ij}=H_{i\ell}=1|\mathcal{G}^{\true}]\cdot\int_{0}^{1}\PP\left[\sqrt{X_{j}'},\sqrt{Y_{\ell}'}\in I_{t}\middle|G_{ij}=H_{i\ell}=1, \mathcal{G}^{\true}\right]\d t.\label{eq:oneventGH1}
\end{align}
On the event $\{G_{ij}=H_{i\ell}=1\}$, we can decompose
$$(X_{j},Y_{\ell})=\left(\left(1+\sum\nolimits_{r\notin\{i,j,\ell\}}G_{jr}\right)+G_{j\ell},\;\left(1+\sum\nolimits_{r\notin\{i,j,\ell\}}H_{\ell r}\right)+H_{j\ell}\right).$$
Note that conditional on $\mathcal{G}^{\true}$, the random variables $\left(1+\sum\nolimits_{r\notin\{i,j,\ell\}}G_{jr}\right)$ and $\left(1+\sum\nolimits_{r\notin\{i,j,\ell\}}H_{\ell r}\right)$, the random vector $(G_{j\ell},H_{j\ell})$ and the event $\{G_{ij}=H_{i\ell}=1\}$ are mutually independent. Moreover, by the symmetry between $G$ and $H$ in the definition of $\mathcal{G}^{\true}$ (see \Cref{rmk:tp2rd}), the conditional distribution of $\left(1+\sum\nolimits_{r\notin\{i,j,\ell\}}H_{\ell r}\right)$ is equal to the distribution of $Z_{2}$ as well. In addition, since $G_{j\ell}\neq H_{j\ell}$ at $\omega$, we have (see \Cref{rmk:tp2rd})
\[\PP[G_{j\ell}=1,H_{j\ell}=0| \mathcal{G}^{\true}](\omega)=\PP[G_{j\ell}=0,H_{j\ell}=1|\mathcal{G}^{\true}](\omega)=\frac{1}{2},\]
meaning that the conditional distribution of $(G_{j\ell},H_{j\ell})$ is the same as the distribution of $(Z_{0},1-Z_{0})$.  
Therefore, by \Cref{def:aliasvar}, the joint distribution of $(X_{j},Y_{\ell})$ 
conditional on both $\mathcal{G}^{\true}$ and the event $\{G_{ij}=G_{i\ell}=1\}$ is the same as the distribution of $(Z_{1}+Z_{0},Z_{2}+1-Z_{0})$. So 
\begin{align}
&\quad\int_{0}^{1}\PP\left[\sqrt{X_{j}'},\sqrt{Y_{\ell}'}\in I_{t}\middle|G_{ij}=H_{i\ell}=1, \mathcal{G}^{\true}\right](\omega)\d t\nonumber\\
&=\int_{0}^{1}\PP\left[\sqrt{\mathsf{Trc}_{1}(Z_{1}+Z_{0})},\sqrt{\mathsf{Trc}_{2}(Z_{2}+1-Z_{0})}\in I_{t}\right]\d t\nonumber\\
&=\E\left[\mu\left(t\in[0,1):\sqrt{\mathsf{Trc}_{1}(Z_{1}+Z_{0})},\sqrt{\mathsf{Trc}_{2}(Z_{2}+1-Z_{0})}\in I_{t}\right)\right]&(\text{by Fubini's theorem})\nonumber\\
&=\E\left[g\left(\sqrt{\mathsf{Trc}_{1}(Z_{1}+Z_{0})}-\sqrt{\mathsf{Trc}_{2}(Z_{2}+1-Z_{0})}\right)\right]&(\text{by \Cref{lem:rationaleg}}).
\label{eq:ridofGH1}
\end{align}
The conclusion follows by combining \eqref{eq:oneventGH1} and \eqref{eq:ridofGH1}.
\end{proof}

Combining the previous two lemmas, we are able to convert $T_{2}(j,\ell)$ to some quantity completely living in the alias world.

\begin{lemma}\label{lem:turnT2intoK}
Assume that $L\geq 5$. The value of $T_{2}(j,\ell)$ at $\omega$ is at most
$$\frac{1}{2}\E\left|K_{00}+K_{11}-K_{10}-K_{01}\right|+\frac{4}{L}\cdot\PP(\Lambda^{c}).$$
\end{lemma}
\begin{proof}
By \Cref{lem:modelingAA,lem:modelingAB} and the fact that 
\begin{equation}\label{eq:neighborhoodsymmetry}
\PP[G_{ij}=G_{i\ell}=1|\mathcal{G}^{\true}]=\PP[G_{ij}=H_{i\ell}=1|\mathcal{G}^{\true}]\leq 1,
\end{equation}
 we know that the value of $T_{2}(j,\ell)$ at $\omega$ is
\begin{align}
&\quad\int_{0}^{1}\PP[A_{ij}'(t)=A_{i\ell}'(t)=1|\mathcal{G}^{\true}](\omega)\d t-\int_{0}^{1}\PP[A_{ij}'(t)=B_{i\ell}'(t)=1|\mathcal{G}^{\true}](\omega)\d t\nonumber\\
&\leq \left|\E\left[g\left(\sqrt{\mathsf{Trc}_{1}(Z_{1}+Z_{0})}-\sqrt{\mathsf{Trc}_{2}(Z_{2}+Z_{0})}\right)\right]-\E\left[g\left(\sqrt{\mathsf{Trc}_{1}(Z_{1}+Z_{0})}-\sqrt{\mathsf{Trc}_{2}(Z_{2}+1-Z_{0})}\right)\right]\right|\nonumber\\
&\leq \frac{1}{2}\E\left|g\left(\sqrt{\mathsf{Trc}_{1}(Z_{1})}-\sqrt{\mathsf{Trc}_{2}(Z_{2})}\right)+g\left(\sqrt{\mathsf{Trc}_{1}(Z_{1}+1)}-\sqrt{\mathsf{Trc}_{2}(Z_{2}+1)}\right)-\right.\nonumber\\
&\qquad\qquad\left.g\left(\sqrt{\mathsf{Trc}_{1}(Z_{1}+1)}-\sqrt{\mathsf{Trc}_{2}(Z_{2})}\right)-g\left(\sqrt{\mathsf{Trc}_{1}(Z_{1})}-\sqrt{\mathsf{Trc}_{2}(Z_{2}+1)}\right)\right|.\label{eq:gofTrc}
\end{align}
From \Cref{lem:ridofTrc} and \Cref{def:defofK} we know that for any $\beta_{1},\beta_{2}\in\{0,1\}$, on $\Lambda$ we always have
$$g\left(\sqrt{\mathsf{Trc}_{1}(Z_{1}+\beta_{1})}-\sqrt{\mathsf{Trc}_{2}(Z_{2}+\beta_{2})}\right)=K_{\beta_{1}\beta_{2}}.$$
In addition, by definition of the function $g$, on $\Lambda^{c}$ we always have 
$$\left|g\left(\sqrt{\mathsf{Trc}_{1}(Z_{1}+\beta_{1})}-\sqrt{\mathsf{Trc}_{2}(Z_{2}+\beta_{2})}\right)\right|\leq \frac{2}{L}.$$
Plugging the preceding two displays into \eqref{eq:gofTrc}, it follows that the value of $T_{2}(j,\ell)$ at $\omega$ is at most
\[\frac{1}{2}\E|K_{00}+K_{11}-K_{10}-K_{01}|+\frac{4}{L}\cdot\PP(\Lambda^{c}).\qedhere\]
\end{proof}

\subsubsection{Analyzing the Alias}\label{subsubsec:aliasanal}

We now return to the alias world to analyze the quantity $K_{00}+K_{11}-K_{10}-K_{01}$.

\begin{definition}[Events]
Let $Q=\mathbb{Z}\cup(\mathbb{Z}+2/L)\cup(\mathbb{Z}-2/L)$ be a discrete set of real numbers. Let $\Lambda_{1}$ be the event that the interval 
$$\left[\sqrt{Z_{1}}-\sqrt{Z_{2}},\sqrt{Z_{1}+1}-\sqrt{Z_{2}}\right)$$
has a nonempty intersection with $Q$. Let $\Lambda_{2}$ be the event that the interval
$$\left[\sqrt{Z_{1}}-\sqrt{Z_{2}+1},\sqrt{Z_{1}}-\sqrt{Z_{2}}\right)$$
has a nonempty intersection with $Q$.
\end{definition}
\begin{lemma}\label{lem:0controlK}
On $\Lambda_{1}^{c}\cap \Lambda_{2}^{c}$, we have $K_{00}+K_{11}-K_{10}-K_{01}=0$.
\end{lemma}
\begin{proof}
On $\Lambda_{1}^{c}\cap \Lambda_{2}^{c}$, the function $g$ is linear on the interval $\left[\sqrt{Z_{1}}-\sqrt{Z_{2}+1},\sqrt{Z_{1}+1}-\sqrt{Z_{2}}\right)$. The conclusion then follows directly from \Cref{def:defofK}.
\end{proof}
\begin{lemma}\label{lem:precontrolK}
Assume $\E[Z_{1}]\geq 2\log^{2}n-1$. Then on $\Lambda$ we have
$$\sqrt{Z_{1}+1}-\sqrt{Z_{1}}\leq \frac{1}{\sqrt{\E[Z_{1}]+1}}.$$
\end{lemma}
\begin{proof}
We have 
$$\left(\sqrt{Z_{1}+1}-\sqrt{Z_{1}}\right)^{2}\leq \frac{1}{4Z_{1}}\leq \frac{1}{4\left(\E[Z_{1}]+1-\log n\sqrt{\E[Z_{1}]-1}\right)}\leq \frac{1}{\E[Z_{1}]+1},$$
where in the second-to-last inequality above we used the condition of $\Lambda$ (\Cref{def:trueLambda}) and in the last inequality we used the assumption $\E[Z_{1}]\geq 2\log^{2}n-1$.
\end{proof}
\begin{lemma}\label{lem:probcontrolK}
Assume $\E[Z_{1}]\geq 2\log^{2}n-1$. Then we have 
\begin{equation}\label{eq:lambda1bd}
\PP(\Lambda_{1}\cap \Lambda)\leq \frac{50\log n}{\sqrt{\Var[Z_{1}]}}.
\end{equation}
\end{lemma}
\begin{proof}
Recall that $Z_{1}$ and $Z_{2}$ are independent. Let us pick an arbitrary integer $y$ and condition on the event $\{Z_{2}=y\}$. 
On the event $\{Z_{2}=y\}$, $\Lambda_{1}\cap \Lambda$ holds only if the intersection
$$\left[\sqrt{Z_{1}},\sqrt{Z_{1}+1}\right)\cap(Q+\sqrt{y})\cap\left[\sqrt{\E[Z_{1}]+1-\log n\sqrt{\E[Z_{1}]-1}},\sqrt{\E[Z_{1}]-1+\log n\sqrt{\E[Z_{1}]-1}}\right]$$
is nonempty. Using the assumption $\E[Z_{1}]\geq 2\log^{2}n-1$, we get from a straightforward computation that
$$\sqrt{\E[Z_{1}]-1+\log n\sqrt{\E[Z_{1}]-1}}-\sqrt{\E[Z_{1}]+1-\log n\sqrt{\E[Z_{1}]-1}}\leq 2\log n.$$
Since each interval of the form $[x,x+1)$ contains at most 3 points of the set $(Q+\sqrt{y})$, it follows that the set
$$(Q+\sqrt{y})\cap\left[\sqrt{\E[Z_{1}]+1-\log n\sqrt{\E[Z_{1}]-1}},\sqrt{\E[Z_{1}]-1+\log n\sqrt{\E[Z_{1}]-1}}\right]$$
has cardinality at most $6\log n+3$. So there are at most $6\log n+3$ possible values of $Z_{1}$ that make $\Lambda_{1}\cap\Lambda$ hold. 
Recall from \Cref{def:aliasvar} that $Z_{1}$ is identically distributed with a sum of independent Bernoulli random variables. By Berry-Esseen theorem (\Cref{cor:B-E}), for any fixed integer $x$,   
$$\PP[Z_{1}=x]\leq \frac{2}{\sqrt{\Var[Z_{1}]}}.$$
Adding up at most $6\log{n}+3$ such inequalities, we have 
$$\PP\left[\Lambda_{1}\cap \Lambda\middle| Z_{2}=y\right]\leq \frac{2}{\sqrt{\Var[Z_{1}]}}(6\log n+3)\leq \frac{50\log n}{\sqrt{\Var[Z_{1}]}}.$$
The conclusion \eqref{eq:lambda1bd} then follows since $y$ is arbitrary.
\end{proof}
Finally, we arrive at an upper bound on $\E|K_{00}+K_{11}-K_{10}-K_{01}|$.
\begin{lemma}\label{lem:controlK}
Assume $\E[Z_{1}],\E[Z_{2}]\geq 2\log^{2}n-1$. Then we have
$$\E\left|K_{00}+K_{11}-K_{10}-K_{01}\right|\leq \frac{100\log n}{\sqrt{\E[Z_{1}]+1}}\left(\frac{1}{\sqrt{\Var[Z_{1}]}}+\frac{1}{\sqrt{\Var[Z_{2}]}}\right).$$
\end{lemma}
\begin{proof}
Using the fact that $g$ is $1$-Lipschitz and then  \Cref{lem:precontrolK}, we have
$$|K_{00}+K_{11}-K_{10}-K_{01}|\leq|K_{00}-K_{10}|+|K_{11}-K_{01}| \leq 2\left|\sqrt{Z_{1}+1}-\sqrt{Z_{1}}\right|\leq \frac{2}{\sqrt{\E[Z_{1}]+1}}.$$
Since $\E[Z_{1}]\geq 2\log^{2}n-1$, we know from \Cref{lem:probcontrolK} that $\PP(\Lambda_{1}\cap\Lambda)\leq 50\log n\left(\Var[Z_{1}]\right)^{-1/2}$. By symmetry, since $\E[Z_{2}]\geq 2\log^{2}n-1$ we also have $\PP(\Lambda_{2}\cap\Lambda)\leq 50\log n\left(\Var[Z_{2}]\right)^{-1/2}$.
It then follows from \Cref{lem:0controlK} that
\begin{align*}
&\quad\E\left|K_{00}+K_{11}-K_{10}-K_{01}\right|\\
&\leq 
\E\Big[\mathbbm{1}\{\Lambda_{1}\cap\Lambda\}\left|K_{00}+K_{11}-K_{10}-K_{01}\right|\Big]+\E\Big[\mathbbm{1}\{\Lambda_{2}\cap\Lambda\}\left|K_{00}+K_{11}-K_{10}-K_{01}\right|\Big]\\
&\leq \frac{50\log n}{\sqrt{\Var[Z_{1}]}}\cdot\frac{2}{\sqrt{\E[Z_{1}]+1}}+\frac{50\log n}{\sqrt{\Var[Z_{2}]}}\cdot\frac{2}{\sqrt{\E[Z_{1}]+1}}.\qedhere
\end{align*}
\end{proof}

\subsubsection{Finishing up}

Armed with all ingredients we need, we can finally provide an upper bound on $T_{2}(j,\ell)$.

\begin{lemma}\label{lem:T2bound}
Assume that $(G,H)$ follows an  $(\alpha,4\alpha^{-1}\log^{2}n,\delta)$-edge-correlated distribution, and that $L \gtrsim \log n$. Then
the value of $T_{2}(j,\ell)$ at $\omega$ is at most 
$$200\alpha^{-3/2}\log n\left(\frac{1}{d_{j}}+\frac{1}{d_{\ell}}\right)+\exp(-\Omega(\log^{2}n)).$$
\end{lemma}
\begin{proof}
Since $\omega\in\Gamma^{\true}\cap\Pi^{\true}$, we know that the value of $\E[X_{j}|\mathcal{G^{\true}}]$ at $\omega$ is at least
$R_{j}^{\true}\geq \frac{\alpha}{2}d_{j}\geq 2\log ^{2}n$.
So (recall \eqref{eq:ZandX}) we have $\E[Z_{1}]\geq \E[X_{j}|\mathcal{G^{\true}}](\omega)-1\geq 2\log^{2}n-1$. By symmetry, we also have $\E[Z_{2}]\geq 2\log^{2}n-1$. We can now combine \Cref{prop:trueLambda} with \Cref{lem:turnT2intoK,lem:controlK} to conclude that the value of $T_{2}(j,\ell)$ at $\omega$ is at most
\begin{equation}\label{eq:prelimT2bound}
\frac{50\log n}{\sqrt{\E[Z_{1}]+1}}\left(\frac{1}{\sqrt{\Var[Z_{1}]}}+\frac{1}{\sqrt{\Var[Z_{2}]}}\right)+\exp(-\Omega(\log^{2}n)).
\end{equation}
By \Cref{def:aliasvar} we know that $\Var[Z_{1}]$ is at least the value of
$$\sum_{r \not\in J^{\true}\cup\{i\}}p_{jr}(1-p_{jr}) \geq \alpha\sum_{r \not\in J^{\true}\cup\{i\}}p_{jr} = \alpha R_{j}^{\true} \geq \frac{\alpha^{2}d_{j}}{2}.$$
So plugging into \eqref{eq:prelimT2bound}, the value of $T_{2}(j,\ell)$ at $\omega$ is at most
\begin{align*}
&\quad\frac{50\log n}{\sqrt{\alpha d_{j}/2}}\left(\frac{1}{\sqrt{\alpha^{2}d_{j}/2}}+\frac{1}{\sqrt{\alpha^{2}d_{\ell}/2}}\right)+\exp(-\Omega(\log^{2}n))\\
&\leq 200\alpha^{-3/2}\log n\left(\frac{1}{d_{j}}+\frac{1}{d_{\ell}}\right)+\exp(-\Omega(\log^{2}n)),
\end{align*}
where the AM-GM inequality was used on the last line.
\end{proof}
\subsection{Proof of \Cref{lem:truepair}}
Now we tidy up the expectation analysis in \Cref{subsec:trueexpectation,subsec:upperbonT2} and conclude with a proof of \Cref{lem:truepair}.
\begin{lemma}[Expectation]\label{lem:trueexpectation}
Assume that $(G,H)$ follows an  $(\alpha,4\alpha^{-1}\log^{2}n,\delta)$-edge-correlated distribution, and that $L \gtrsim \log n$. On $\Gamma^{\true}\cap\Pi^{\true}$, we have
$$\E\left[D'(i,i)^{2}\middle|\mathcal{G}^{\true}\right]\leq 10^{4}d_{i}\left(\frac{\delta}{L}+\sqrt{\alpha^{-1}\delta}+\alpha^{-3/2}\delta\log n\right)+\exp(-\Omega(\log^{2}n)).$$
\end{lemma}
\begin{proof}
Combining \Cref{prop:nilT2} with \Cref{lem:T2bound}, on $\Gamma^{\true} \cap \Pi^{\true}$ we have
\begin{align*}
\sum_{\{j,\ell\}\subset J^{\true}}T_{2}(j,\ell)&\leq 200\alpha^{-3/2}\log n\sum_{\{j,\ell\}\subset J^{\true}}\left(\frac{1}{d_{j}}+\frac{1}{d_{\ell}}\right)\mathbbm{1}\{G_{j\ell} \neq H_{j\ell}\}+\sum_{\{j,\ell\}\subset J^{\true}}\exp(-\Omega(\log^{2}n))\\
&\leq 2000\alpha^{-3/2}\delta d_{i}\log n+\exp(-\Omega(\log^{2}n)),
\end{align*}
where in the last inequality we used the condition $S^{\true}\leq 10\delta d_{i}$, as guaranteed by $\Pi^{\true}$.
The result follows by combining the above with \Cref{lem:turntrueexptoT,lem:T1bound}.
\end{proof}
\begin{proof}[Proof of \Cref{lem:truepair}]
We first pick $c_{1}'=4/\alpha$. By \Cref{lem:trueexpectation} and Cauchy-Schwarz, we have on $\Gamma^{\true}\cap \Pi^{\true}$
$$\E\left[\frac{D'(i,i)}{\sqrt{d_{i}}}\middle|\mathcal{G}^{\true}\right]\leq 100\sqrt{\frac{\delta}{c_{3}\log n}}+100\alpha^{-1/4}\delta^{1/4}+100 \alpha^{-3/4}\sqrt{\delta\log n}+\exp(-\Omega(\log^{2}n)).$$
Combining the above with \Cref{lem:trueconcentration}, on $\Gamma^{\true}\cap\Pi^{\true}$
$$\PP\left[\frac{D'(i,i)}{\sqrt{d_{i}}}\geq c_{5}\sqrt{\delta\log n}+c_{6}\delta^{1/4}+\frac{100}{c_{3}\sqrt{\log n}}\middle|\mathcal{G}^{\true}\right]\leq O(n^{-3}),$$
for some constants $c_{5},c_{6}$ depending on $\alpha$ and $c_{3}$. Since $\PP(\Gamma^{\true}\cap\Pi^{\true})\geq 1-\exp(-\Omega(\log^{2}n))$ by \Cref{prop:trueGamma,prop:trueTheta}, and since $\PP[D'(i,i)\neq D(i,i)]\leq \exp(-\Omega(\log^{2}n))$ by \Cref{lem:DequalsDprime}, we conclude that
\[\PP\left[\frac{D(i,i)}{\sqrt{d_{i}}}\geq c_{5}\sqrt{\delta\log n}+c_{6}\delta^{1/4}+\frac{100}{c_{3}\sqrt{\log n}}\right]\leq O(n^{-3})+\exp(-\Omega(\log^{2}n))= O(n^{-3}).\qedhere\]
\end{proof}

\bibliographystyle{plain}
\bibliography{reference}
\appendix
\section{Probabilistic Inequalities}
The following concentration inequalities are used many times in this paper.

\begin{theorem}[Hoeffding's inequality, \mbox{\cite[Theorem D.2]{mohri2018foundations}}]\label{thm:Hoeffding}
Let $X_1, \dots, X_n$ be independent random variables with $X_{i}$ taking values in $[a_{i},b_{i}]$ for all $i\in [n]$. Then for any $\varepsilon\geq 0$ we have
$$\PP\left[\left|\sum_{i=1}^{n}X_i-\sum_{i=1}^{n}\E[X_{i}]\right| \geq \varepsilon\right] \leq 2\exp\left(-\frac{2\varepsilon^{2}}{\sum_{i=1}^{n}(b_i-a_i)^{2}}\right).$$
\end{theorem}

\begin{theorem}[Multiplicative Chernoff bounds, \mbox{\cite[Theorem D.4]{mohri2018foundations}}]\label{thm:Chernoff}
Let $X_1, \dots, X_n$ be independent random variables taking values in $[0,1]$. Let $S=\sum_{i=1}^{n}X_i$ and $\mu=\E[S]$. For any $\varepsilon\geq 0$ we have
$$\PP[S \leq (1-\varepsilon)\mu] \leq \exp\left(-\frac{1}{2}\varepsilon^{2}\mu \right)$$
$$\PP[S \geq (1+\varepsilon)\mu] \leq \exp\left(-\frac{\varepsilon^{2}}{2+\varepsilon}\mu \right).$$
\end{theorem}

Both \Cref{thm:Hoeffding} and \Cref{thm:Chernoff} can be derived from moment generating functions, but they are tailored to different situations. In particular, \Cref{thm:Chernoff} is finer than \Cref{thm:Hoeffding} when $\mu$ is small, while in \Cref{thm:Hoeffding} the ranges of the random variables are more flexible. 

In this paper we often use \Cref{thm:Chernoff} in the following form:

\begin{corollary}\label{cor:Chernoff}
Let $X_1, \dots, X_n$ be independent random variables taking values in $[0,1]$. Let $S=\sum_{i=1}^{n}X_i$. If $a,b$ are real numbers such that $a\leq b\leq \E[S]$, then
\begin{equation}\label{eq:Chernoffleft}
\PP[S\leq a]\leq \exp\left(-\frac{1}{2}(b-a)\cdot\frac{b-a}{b}\right).
\end{equation}
If $a,b$ are real numbers such that $\E[S]\leq b\leq a$, then
\begin{equation}\label{eq:Chernoffright}
\PP[S\geq a]\leq \exp\left(-\frac{1}{3}(a-b)\cdot\min\left\{\frac{a-b}{b},1\right\}\right).
\end{equation}
\end{corollary}

\begin{proof}
For \eqref{eq:Chernoffleft}, note that the right hand side of \eqref{eq:Chernoffleft} is monotone decreasing in $b$. So we may assume $b=\E[S]$. Then \eqref{eq:Chernoffleft} becomes a restatement of the first inequality of \Cref{thm:Chernoff}.

For \eqref{eq:Chernoffright}, note that the right hand side of \eqref{eq:Chernoffright} is monotone increasing in $b$. So we may assume $b=\E[S]$. Since $\frac{1}{3}\varepsilon\cdot\min\{\varepsilon,1\}\leq \varepsilon^{2}/(2+\varepsilon)$, \eqref{eq:Chernoffright} becomes a direct consequence of the second inequality of \Cref{thm:Chernoff}.
\end{proof}

McDiarmid's inequality is a generalization of \Cref{thm:Hoeffding}, and is our principal tool for the concentration analysis in \Cref{subsec:fakeconcentrate,subsec:trueconcentration}. 

\begin{theorem}[McDiarmid's inequality, \mbox{\cite[Theorem D.8]{mohri2018foundations}}]\label{thm:Mcdiarmid}
Let $X_1, \dots, X_n$ be independent random variables taking value in a space $\mathcal{X}$. Let $f: \mathcal{X}^{n} \rightarrow \mathbb{R}$ be a measurable function. Assume that for any $(x_1, \dots, x_n) \in \mathcal{X}^{n}$, 
\begin{equation}\label{eq:marginaldiff}
\sup_{x_{i}^{\mathsf{new}} \in \mathcal{X}}|f(x
_1, \dots, x_{i-1}, x_i, x_{i+1}, \dots, x_n)-f(x
_1, \dots, x_{i-1}, x_{i}^{\mathsf{new}}, x_{i+1}, \dots, x_n)| \leq c_i.
\end{equation}
That is, substituting the value of the $i$-th coordinate $x_i$ changes the value of function $f$ by at most $c_i$. Let $X=(X_1, \dots, X_n)$. Then for any $\Delta\geq 0$, with probability at least $1-2e^{-\Delta}$,
$$\left|f(X)-\E[f(X)]\right| \leq \sqrt{\frac{\Delta}{2}\sum_{i=1}^{n}c_{i}^{2}}.$$
\end{theorem}

The following classical result is useful in \Cref{lem:fakeproblowerb,lem:probcontrolK}.

\begin{theorem}[Berry-Esseen theorem, \cite{berry1941accuracy,esseen1942liapunov}]\label{thm:B-E}
Let $X_1, \dots, X_n$ be independent random variables with $\E[X_i]=0$, $\E[X_i^{2}]=\sigma_{i}^{2}>0$, and $\E[|X_i^{3}|]=\rho_i<\infty$. Let 
$$S_n=\frac{X_1+X_2+\cdots+X_n}{\sqrt{\sigma_1^2+\sigma_2^2+\cdots+\sigma_n^2}}$$
be the normalized $n$-th partial sum. Denote by $F_n$ the cumulative distribution function of $S_n$ and by $\Phi$ the cumulative distribution function of the standard normal distribution. Then we have
$$\sup_{x \in \mathbb{R}}|F_n(x)-\Phi(x)| \leq \left(\sum_{i=1}^{n}\sigma_i^2\right)^{-3/2} \cdot \sum_{i=1}^{n}\rho_i.$$
\end{theorem}
The proof of a sharper version of \Cref{thm:B-E} can be found in \cite{shevtsova2010improvement}. In this paper \Cref{thm:B-E} is used in the following form:
\begin{corollary}\label{cor:B-E}
Let $X$ be a sum of independent Bernoulli random variables. Let $X'$ follow a normal distribution of mean $\E[X]$ and variance $\Var[X]$. Then for any real numbers $a$ and $b$, we have
$$\left|\PP[a\leq X\leq b]-\PP[a\leq X'\leq b]\right|\leq \frac{2}{\sqrt{\Var[X]}}.$$
\end{corollary}
\begin{proof}
Write $X=\E[X]+X_{1}+\dots+X_{n}$, where $X_{i}$ are independent ``centered'' Bernoulli variables of the form 
$$\PP[X_{i}=1-p_{i}]=p_{i}\text{ and }\PP[X_{i}=-p_{i}]=1-p_{i}.$$
Let $\sigma_{i}^{2}=\E[X_{i}^{2}]$ and $\rho_{i}=\E[|X_{i}^{3}|]$. Since $|X_{i}|\leq 1$ we clearly have $\rho_{i}\leq \sigma_{i}^{2}$. 

Let $F_{n}$ be the cumulative distribution function of the random variable $$S_{n}:=\frac{X_{1}+\dots+X_{n}}{\sqrt{\sigma_{1}^{2}+\dots+\sigma_{n}^{2}}}=\frac{X-\E[X]}{\sqrt{\Var[X]}},$$ and let $\Phi$ be the cumulative distribution of the standard normal distribution. We have
\begin{align*}
\left|\PP[X\leq b]-\PP[X'\leq b]\right|&=\left|F_{n}\left(\frac{b-\E[X]}{\Var[X]}\right)-\Phi\left(\frac{b-\E[X]}{\Var[X]}\right)\right|\\
&\leq\frac{\rho_{1}+\dots+\rho_{n}}{(\sigma_{1}^{2}+\cdots+\sigma_{n}^{2})^{3/2}}&(\text{by \Cref{thm:B-E}})\\
&\leq \frac{1}{\sqrt{\sigma_{1}^{2}+\cdots+\sigma_{n}^{2}}}=\frac{1}{\sqrt{\Var[X]}}&(\text{since }\rho_{i}\leq \sigma_{i}^{2}).
\end{align*}
Similarly
$$
\left|\PP[X<a]-\PP[X'<a]\right|=\lim_{x\rightarrow a^{-}}\left|\PP[X\leq x]-\PP[X'\leq x]\right|\leq \frac{1}{\sqrt{\Var[X]}}.
$$
Therefore
\[\left|\PP[a\leq X\leq b]-\PP[a\leq X'\leq b]\right|\leq \left|\PP[X\leq b]-\PP[X'\leq b]\right|+\left|\PP[X<a]-\PP[X'<a]\right|\leq\frac{2}{\sqrt{\Var[X]}}.\qedhere\]
\end{proof}

\section{Symmetric $\{-1,0,1\}$ Variables}

This appendix serves two purposes. Firstly, it provides the lemma required in the proof of \Cref{lem:fakeexpectation}. Secondly, it lays the groundwork for \Cref{sec:gaussian}. The focal point of this appendix is the following simple notion that somehow unifies these two objectives.

\begin{definition}
A random variable $X$ is said to be a symmetric $\{-1,0,1\}$ variable if $X$ takes value in the set $\{-1,0,1\}$ and $\PP[X=-1]=\PP[X=1]$.
\end{definition}

The lemma below demonstrates how symmetric $\{-1,0,1\}$ variables effectively symmetrize Bernoulli variables.

\begin{lemma}\label{lem:Bern_to_sym}
Let $X_{1},\dots,X_{n}$ be independent Bernoulli variables, and for $i\in[n]$ let $p_{i}=\PP[X_{i}=1]$. Suppose $Y_{1},\dots,Y_{n}$ are independent symmetric $\{-1,0,1\}$ variables with $\PP[Y_{i}=1]=\min\{p_{i},1-p_{i}\}$. Then
$$\inf_{x\in\mathbb{R}}\E\left|\sum_{i=1}^{n}X_{i}-x\right|\geq \frac{1}{2}\E\left|\sum_{i=1}^{n}Y_{i}\right|.$$
\end{lemma}
\begin{proof}
For each $i\in[n]$, let $W_{i}$ be a Bernoulli random variable with $\PP[W_{i}=1]=2\min\{p_{i},1-p_{i}\}$, and let $U_{i}$ be a random variable taking value in $\{-1,1\}$ with $\PP[U_{i}=1]=\frac{1}{2}$. Let $W_{1},\dots,W_{n},U_{1},\dots,U_{n}$ be mutually independent. Finally, let $u_{i}$ be a constant defined by
$$u_{i}=\begin{cases}
-1 &\text{if }p_{i}\leq \frac{1}{2},\\
1 &\text{if }p_{i}>\frac{1}{2}.
\end{cases}$$
It is then easy to observe that for each $i\in [n]$, the random variables $Y_{i}$ and $W_{i}U_{i}$ are identically distributed. In addition, $2X_{i}-1$ and $W_{i}U_{i}+(1-W_{i})u_{i}$ are identically distributed. So
\begin{equation}\label{eq:medianX}
\inf_{x\in \mathbb{R}}\E\left|\sum_{i=1}^{n}X_{i}-x\right|=\frac{1}{2}\inf_{x\in\mathbb{R}}\E\left|\sum_{i=1}^{n}(2X_{i}-1)-x\right|=
\frac{1}{2}\inf_{x\in\mathbb{R}}\E\left|\sum_{i=1}^{n}\Big(W_{i}U_{i}+(1-W_{i})u_{i}\Big)-x\right|.
\end{equation}
It is easy to simplify the right hand side of \eqref{eq:medianX} if $W_{1},\dots,W_{n}$ are fixed. In fact, since each variable $U_i$ is symmetric around 0, the conditional distribution of $\sum_{i=1}^{n}W_{i}U_{i}$ (given $W_1, \ldots, W_n$) has median 0. Recall that the $L^{1}$-norm of a translated variable $Z-z$ is always minimized when $z$ is the median of the variable $Z$. So for any $x\in\mathbb{R}$, we have
$$\E\left[\left|\sum_{i=1}^{n}W_{i}U_{i}+\left(\sum_{i=1}^{n}(1-W_{i})u_{i}-x\right)\right|\mid W_{1},\dots,W_{n}\right] \geq \E\left[\left|\sum_{i=1}^{n}W_{i}U_{i}\right|\mid W_{1},\dots,W_{n}\right].$$
Taking expectation and then infimum over $x\in\mathbb{R}$ yields
\begin{equation}\label{eq:medianY}
\inf_{x\in\mathbb{R}}\E\left|\sum_{i=1}^{n}\Big(W_{i}U_{i}+(1-W_{i})u_{i}\Big)-x\right|\geq \E\left|\sum_{i=1}^{n}W_{i}U_{i}\right|=\E\left|\sum_{i=1}^{n}Y_{i}\right|.
\end{equation}
Combining \eqref{eq:medianX} and \eqref{eq:medianY} concludes the proof.
\end{proof}
We will now delve deeper into the study of sums of independent symmetric $\{-1,0,1\}$ variables. In particular, we examine the $L^{1}$-norm of the sum as a function of the probability parameters associated with each individual symmetric $\{-1,0,1\}$ variable.
\begin{definition}\label{def:Bfunction}
Let $n$ be any positive integer. We define a function $f_{n}:[0,\frac{1}{2}]^{n}\rightarrow\mathbb{R}$ as follows. For $p_{1},\dots,p_{n}\in[0,\frac{1}{2}]$, let $X_{1},\dots,X_{n}$ be independent symmetric $\{-1,0,1\}$ variables with $\PP[X_{i}=1]=p_{i}$ for each $i\in [n]$. Define
$$f_{n}(p_{1},\dots,p_{n}):=\E\left|\sum_{i=1}^{n}X_{i}\right|.$$
\end{definition}

\begin{lemma}\label{lem:control_f}
Suppose $p_{1},\dots,p_{n} \in [0,1/2]$ and let $S=p_{1}+\dots+p_{n}$. Then
$$f_{n}(p_{1},\dots,p_{n})\geq \frac{2S}{\sqrt{6S+1}}.$$
\end{lemma}
\begin{proof}
Let $X_{1},\dots,X_{n}$ be independent symmetric $\{-1,0,1\}$ variables with $\PP[X_{i}=1]=p_{i}$ and let 
$X:=\sum_{i=1}^{n}X_{i}$.
We now compute 
\begin{equation}\label{eq:2nd_moment}
\E[X^{2}]=\sum_{i=1}^{n}\E[X_{i}^{2}]=2\sum_{i=1}^{n}p_{i}=2S.
\end{equation}
In addition, we have by independence
\begin{align}
\E[X^{4}]=&\E\left[\left(\sum_{i=1}^{n}X_{i}\right)^{4}\right]
=\sum_{i=1}^{n}\E[X_{i}^{4}]+3\sum_{i=1}^{n}\sum_{j \neq i}\E[X_{i}^{2}]\E[X_{j}^{2}]\nonumber\\
=&2\sum_{i=1}^{n}p_{i}+3\sum_{i=1}^{n}2p_{i}(2S-2p_{i})
\leq 2S+12S^{2}.\label{eq:4th_moment}
\end{align}
Furthermore, by H\"{o}lder's inequality we have
$$\E[X^{4}]\cdot \left(\E|X|\right)^{2}\geq \left(\E[X^{2}]\right)^{3}.$$
Combined with \eqref{eq:2nd_moment} and \eqref{eq:4th_moment}, it yields the desired conclusion.
\end{proof}

In \Cref{lem:fakeexpectation}, we used \Cref{lem:control_f} in the following form.
\begin{corollary}\label{cor:control_f}
Suppose $X_{1},\dots,X_{n}$ are independent Bernoulli random variables. For each $i\in [n]$ suppose $\PP[X_{i}]=p_{i}$ such that $\frac{1}{n}\leq p_{i}\leq 1-\frac{1}{n}$. Then for any random variable $Y$ independent with $(X_{i})_{i\in [n]}$, we have
$$\E\left|\sum_{i=1}^{n}X_{i}+Y\right|\geq \frac{1}{4}\sqrt{\sum_{i=1}^{n}\min\{p_{i},1-p_{i}\}}.$$
\end{corollary}
\begin{proof}
By independence and then using \Cref{lem:Bern_to_sym}, we have
\begin{equation}\label{eq:cor_control_f_1}
\E\left|\sum_{i=1}^{n}X_{i}+Y\right|\geq \inf_{x\in\mathbb{R}}\E\left|\sum_{i=1}^{n}X_{i}-x\right|\geq \frac{1}{2}f_{n}\Big(\min\{p_{1},1-p_{1}\},\dots,\min\{p_{n},1-p_{n}\}\Big).
\end{equation}
Let $S:=\min\{p_{1},1-p_{1}\}+\dots+\min\{p_{n},1-p_{n}\}$. Then $S \geq 1$. By \Cref{lem:control_f} we then have
\begin{equation}\label{eq:cor_control_f_2}
f_{n}\Big(\min\{p_{1},1-p_{1}\},\dots,\min\{p_{n},1-p_{n}\}\Big)\geq \frac{2S}{\sqrt{6S+1}}\geq \frac{2S}{\sqrt{7S}}\geq \frac{1}{2}\sqrt{S}.
\end{equation}
The conclusion then follows directly by combining \eqref{eq:cor_control_f_1} and \eqref{eq:cor_control_f_2}.
\end{proof}

While the absolute lower bound in \Cref{lem:control_f} is sufficient for \Cref{lem:fakeexpectation}, the analysis in \Cref{sec:gaussian} requires a more careful examination on how $f_{n}$ changes with its arguments $p_{1},\dots,p_{n}$. Our next goal is to establish \Cref{lem:compare_f}, a comparison result in this regard.

\begin{lemma}\label{lem:f_monotone}
The function $f_{n}$ is non-decreasing in each of its arguments.
\end{lemma}

\begin{proof}
Since $f_{n}$ is a symmetric function, it suffices to show monotonicity in the first argument. Assume $X_{1},X_{2},\dots,X_{n}$ are independent symmetric $\{-1,0,1\}$ random variables, and let $p_{i}=\PP[X_{i}=1]$ for $i\in [n]$. For each integer $t$, note that
$$\E|X_{1}+t|-|t|=\begin{cases}
2p_{1} &\text{if }t=0,\\
0 &\text{if }t\neq 0.
\end{cases}$$
So by independence, we have
$$\E\left|\sum_{i=1}^{n}X_{i}\right|-\E\left|\sum_{i=2}^{n}X_{i}\right|=2p_{1}\cdot\PP\left[\sum_{i=2}^{n}X_{i}=0\right].$$
That is,
$$f_{n}(p_{1},p_{2},\dots,p_{n})=f_{n}(0,p_{2},\dots,p_{n})+2p_{1}\cdot\PP\left[\sum_{i=2}^{n}X_{i}=0\right],$$
from which we see that $f_{n}$ is non-decreasing in $p_{1}$.
\end{proof}

The following lemma is on a basic property of sums of independent symmetric variables, which will be useful in proving \Cref{lem:compare_f}.

\begin{lemma}\label{lem:control_h}
Suppose $p_{1},\dots,p_{n}\in[0, \frac{1}{4}]$ and let $X_{1},\dots,X_{n}$ be independent symmetric $\{-1,0,1\}$ variables with $\PP[X_{i}=1]=p_{i}$. For any integers $a_{1},\dots,a_{n}$ and any integer $x$, we have 
$$\PP\left[\sum_{i=1}^{n}a_{i}X_{i}=0\right]\geq \PP\left[\sum_{i=1}^{n}a_{i}X_{i}=x\right].$$ 
\end{lemma}
\begin{proof}
Let $Y_{1},\dots,Y_{n},Z_{1},\dots,Z_{n}$ be independent Bernoulli variables with 
$$\PP[Y_{i}=1]=\PP[Z_{i}=1]=\frac{1-\sqrt{1-4p_{i}}}{2}.$$
Then the variables $X_{i}$ and $Y_{i}-Z_{i}$ are identically distributed, for each $i\in [n]$. So for any integer $x$, we have
\begin{align*}
&\quad\PP\left[\sum_{i=1}^{n}a_{i}X_{i}=x\right]=\PP\left[\sum_{i=1}^{n}a_{i}Y_{i}=\sum_{i=1}^{n}a_{i}Z_{i}+x\right]\\
&=\sum_{y\in\mathbb{Z}}\PP\left[\sum_{i=1}^{n}a_{i}Y_{i}=y+x\right]\cdot\PP\left[\sum_{i=1}^{n}a_{i}Z_{i}=y\right]&(\text{by independence})\\
&\leq \frac{1}{2}\sum_{y\in \mathbb{Z}}\PP\left[\sum_{i=1}^{n}a_{i}Y_{i}=y+x\right]^{2}+\frac{1}{2}\sum_{y\in \mathbb{Z}}\PP\left[\sum_{i=1}^{n}a_{i}Z_{i}=y\right]^{2} &(\text{by AM-GM inequality})\\
&=\sum_{y\in\mathbb{Z}}\PP\left[\sum_{i=1}^{n}a_{i}Y_{i}=y\right]\cdot \PP\left[\sum_{i=1}^{n}a_{i}Z_{i}=y\right] &\left(Y_{i}\text{ and }Z_{i}\text{ are identically distributed}\right)\\
&=\PP\left[\sum_{i=1}^{n}a_{i}(Y_{i}-Z_{i})=0\right]=\PP\left[\sum_{i=1}^{n}a_{i}X_{i}=0\right].&\qedhere
\end{align*}
\end{proof}

\begin{lemma}\label{lem:compare_f}
Suppose $p_{1},\dots,p_{n}\in[0,\frac{1}{16}]$. Then
$$f_{n}(4p_{1},\dots,4p_{n})\geq 2f_{n}(p_{1},\dots,p_{n}).$$
\end{lemma}
\begin{proof}
Let $X_{1},\dots,X_{n},Y_{1},\dots,Y_{n}$ be independent symmetric $\{-1,0,1\}$ variables with $\PP[X_{i}=1]=4p_{i}$ and $\PP[Y_{i}=1]=p_{i}$ for each $i\in [n]$. For each $k\in [n]$, observe that for any $t\in \mathbb{Z}$ we have
$$\E|X_{k}+t|-\E|2Y_{k}+t|=\begin{cases}
0 &\text{if }|t|\geq 2,\\
-2p_{k} &\text{if }|t|=1,\\
4p_{k} &\text{if }t=0.
\end{cases}$$
So by independence between the variables, we have
\begin{align}
&\quad\E\left|\sum_{i=1}^{k-1}X_{i}+X_{k}+\sum_{i=k+1}^{n}2Y_{i}\right|-\E\left|\sum_{i=1}^{k-1}X_{i}+2Y_{k}+\sum_{i=k+1}^{n}2Y_{i}\right|\nonumber\\
&=-2p_{k}\cdot\PP\left[\left|\sum_{i=1}^{k-1}X_{i}+\sum_{i=k+1}^{n}2Y_{i}\right|=1\right]+4p_{k}\cdot\PP\left[\sum_{i=1}^{k-1}X_{i}+\sum_{i=k+1}^{n}2Y_{i}=0\right]\geq 0,\label{eq:compare_f}
\end{align}
using \Cref{lem:control_h} in the last step.
Summing over all $k\in [n]$ on both sides of \eqref{eq:compare_f} yields
$$\E\left|\sum_{i=1}^{n}X_{i}\right|\geq \E\left|\sum_{i=1}^{n}2Y_{i}\right|,$$
which is exactly the desired inequality.
\end{proof}

\section{Matching Gaussian Matrices}\label{sec:gaussian}

\iffalse
This appendix focus on a similar but simplified model of the inhomogeneous correlated Erd\H{o}s-R\'{e}nyi graph model in \Cref{thm:graphmatching}, that is, correlated Gaussian matrix model. We also provide a similar theorem to \Cref{thm:graphmatching}. Although the proof for \Cref{thm:graphmatching} is not carried out in exactly the same way as our theorem here, we wish to convey some of the main ideas and intuitions through this appendix.
\fi

In this appendix, we provide a degree-profile based algorithm for matching correlated Gaussian matrices with arbitrary and inhomogeneous means. While mathematically our \Cref{thm:graphmatching} is self-contained from this appendix, we believe that the analysis on the Gaussian case serves as a helpful warm up and illustrates well some of the key intuitions in a simplified context. 

Let $(G,H)$ be a pair of random symmetric $n\times n$ matrices where each pair $(G_{ij},H_{ij})$ is sampled from the bivariate normal distribution
$$\mathcal{N}\left(\begin{bmatrix}\mu_{ij}\\ \mu_{ij}\end{bmatrix},
\begin{bmatrix}
1 & \rho_{ij}\\
\rho_{ij} & 1
\end{bmatrix}
\right)$$
independently, for $1\leq i\leq j\leq n$. Here, $\mu_{ij}$ and $\rho_{ij}$ are arbitrary and unknown real numbers, with the assumption that $\rho_{ij}$ is no less than a given threshold $\rho=1-O(\log^{-2}n)$. We want to find an efficient algorithm which succeeds to recover any unknown permutation $\pi \in \mathfrak S_n$ with high probability when the input is given as $(G_{ij})_{i,j\in [n]}$ and the permuted matrix $(H_{\pi^{-1}(i),\pi^{-1}(j)})_{i,j\in [n]}$.

We again base our algorithm on a balls-into-bins-type distance function $D:[n]\times[n]\rightarrow \mathbb{R}$. For each $i$, we directly use entries in the row $(G_{ij})_{j\in [n]}$ (or correspondingly the row $(H_{ij})_{j\in [n]}$) as ``balls'', as in \cite[Section 2]{DMWX21}; and our ``bins'' will be the intervals $I_{m}=[\frac{m}{L},\frac{m+1}{L})$, for $m\in\mathbb{Z}$. Here $L$ is a fixed positive integer. For each $i,j\in [n]$ we define the Bernoulli variables
$$A_{ij}(m)=\mathbbm{1}\{G_{ij}\in I_{m}\}\text{ and }B_{ij}(m)=\mathbbm{1}\{H_{ij}\in I_{m}\}.$$
We define the distance
\begin{equation}\label{eq:Cdef_of_D}
D(i,k)=\sum_{m\in\mathbb{Z}}\left|\sum_{j=1}^{n}A_{ij}(m)-\sum_{j=1}^{n}B_{kj}(m)\right|.
\end{equation}

We will show in this appendix that the same algorithm as \Cref{alg:graphmatching} but based on the distance $D(\cdot,\cdot)$ in \eqref{eq:Cdef_of_D} suffices for the Gaussian random matrices $(G,H)$. Since our algorithm is permutation-oblivious (see \Cref{prop:algoblivious}), we may assume without loss of generality that $\pi=\mathrm{id}$, where $\mathrm{id}$ is the identity permutation.

\begin{theorem}\label{thm:matrixmatching}
Suppose $L\geq 10^{5}\log n$ and $\rho\geq 1-10^{-8}L^{-2}$. Let $\mathcal{A}$ be the algorithm that outputs $\widehat{\pi}$ where $\widehat{\pi}(i)=\argmin_{k\in [n]}D(i,k)$ for $i \in [n]$ (we use the convention that in the case of multiple minimizers an arbitrary one is picked). Then we have
$$\PP\left[\mathcal{A}\left(\left(G_{ij}\right)_{i,j\in[n]},\;\left(H_{ij}\right)_{i,j\in[n]},\;L\right)=\mathrm{id}\right]\geq 1-O(n^{-1}).$$
\end{theorem}

\Cref{thm:matrixmatching} follows immediately from the next lemma and a simple union bound.

\begin{lemma}\label{lem:matrixmatching}
Suppose $L\geq 10^{5}\log n$ and $\rho\geq 1-10^{-8}L^{-2}$. Let $i,k \in [n]$ be a pair of distinct indices. We have 
$$\PP[D(i,k)<D(i,i)] \leq O(n^{-3}).$$
\end{lemma}

The remaining of this appendix is devoted to proving \Cref{lem:matrixmatching}. We fix the indices $i$ and $k$ throughout. The statement of \Cref{lem:matrixmatching} concerns only three rows $(G_{ij})_{j\in [n]}$, $(H_{ij})_{j\in [n]}$ and $(H_{kj})_{j\in [n]}$ of the matrices. The row $(G_{ij})_{j\in [n]}$ is entry-wise correlated with the row $(H_{ij})_{j\in [n]}$ but is independent with the row $(H_{kj})_{j\in [n]}$, except that there is dependence between $G_{ik}$ and $H_{ki}$. As in \cite[Section 2]{DMWX21}, we may assume $G_{ik}$ and $H_{ki}$ are also independent, since by the triangle inequality this changes $D(i,k)-D(i,i)$ by at most an additive constant. In the following, we will show
\begin{equation}\label{eq:matrixmatching}
\PP\left[D(i,k)- D(i,i)\leq 2\sqrt{n\log n}\right]\leq O(n^{-3})
\end{equation}
under the assumption that $G_{ik}$ and $H_{ki}$ are also independent, which implies \Cref{lem:matrixmatching}.

\begin{definition}\label{def:good}
A pair $(j,m) \in [n]\times \mathbb{Z}$ is said to be ``good'' if $I_{m}=\left[\frac{m}{L},\frac{m+1}{L}\right)$ satisfies 
$$\mu_{ij}-\log^{2/3}n\leq \frac{m}{L}<\frac{m+1}{L}\leq \mu_{ij}+\log^{2/3}n.$$    
\end{definition}

\begin{definition}[Truncation]
We define the following approximation for $D(i, i)$: 
\begin{equation}\label{eq:Cdef_of_Dprime}
D'(i,i):=\sum_{m \in \mathbb{Z}}\left|\sum_{j:(j,m) \text{ good pair}}(A_{ij}(m)-B_{ij}(m))\right|.
\end{equation}
\end{definition}

\begin{lemma}\label{lem:Ctrunc}
We have
$$\PP\left[D'(i,i) \neq D(i,i)\right]=\exp\left(-\Omega\left(\log^{4/3}n\right)\right).$$
\end{lemma}

\begin{proof}
By union bound
\begin{align*}\PP\left[D'(i,i) \neq D(i,i)\right] \leq& \sum_{m \in \mathbb{Z}}\sum_{j=1}^{n}\mathbbm{1}\{(j,m) \text{ not good}\}\left(\PP[A_{ij}(m) \neq 0]+\PP[B_{ij}(m) \neq 0]\right) \\
\leq & \sum_{j=1}^{n}\left(\PP\left[|G_{ij}-\mu_{ij}| \geq \log^{2/3}n\right]+\PP\left[|H_{ij}-\mu_{ij}| \geq \log^{2/3}n\right]\right)\\
=& \exp\left(-\Omega\left(\log^{4/3}n\right)\right).\qedhere
\end{align*}
\end{proof}

We then study the concentration and expectation of the variable $D(i,k)-D'(i,i)$ separately.

\begin{lemma}[Concentration]\label{lem:Cconcentration} 
We have
$$\PP\left[D(i,k)-D'(i,i) \leq \E\left[D(i,k)-D'(i,i)\right]-10\sqrt{n\log n}\right] \leq O(n^{-3}).$$
\end{lemma}
\begin{proof}
We group the $3n$ random variables $G_{i1},\dots,G_{in},H_{i1},\dots,H_{in},H_{k1},\dots,H_{kn}$ into $n$ bundles, the $j$-th of which consists of $G_{ij},H_{ij}$ and $G_{kj}$. Note that these $n$ bundles are mutually independent. We now consider $D(i,k)-D'(i,i)$ as a function of the $n$ bundles and analyze the marginal difference a bundle (as a whole) can make on the quantity $D(i,k)-D'(i,i)$. In the following, we use the superscript $\mathsf{new}$ to denote the new random variables obtained after changing the $j$-th bundle. Note that $A_{ij}(\cdot)$, $A_{ij}(\cdot)^{\mathsf{new}}$, $B_{ij}(\cdot)$ and $B_{ij}(\cdot)^{\mathsf{new}}$ are all indicator functions of singleton sets. We then have by triangle inequality
\begin{align*}
\left|D'(i,i)^{\mathsf{new}}-D'(i,i)\right|&\leq \sum_{r=1}^{n}\sum_{m\in\mathbb{Z}}\Big(\left|A_{ir}(m)^{\mathsf{new}}-A_{ir}(m)\right|+\left|B_{ir}(m)^{\mathsf{new}}-B_{ir}(m)\right|\Big)\\
&\leq\sum_{m\in\mathbb{Z}}\left|A_{ij}(m)^{\mathsf{new}}\right|+\sum_{m\in\mathbb{Z}}\left|A_{ij}(m)\right|+\sum_{m\in\mathbb{Z}}\left|B_{ij}(m)^{\mathsf{new}}\right|+\sum_{m\in\mathbb{Z}}\left|B_{ij}(m)\right|\\
&\leq 1+1+1+1=4.
\end{align*}
Similarly $\left|D(i,k)^{\mathsf{new}}-D(i,k)\right|\leq 4$. By the triangle inequality again, for every $j \in [n]$, the marginal difference on the quantity $D(i,k)-D'(i,i)$ caused by the $j$-th bundle is no more than $8$. By McDiarmid's inequality (\Cref{thm:Mcdiarmid}) we conclude that with probability at least $1-O(n^{-3})$,
\begin{align*}
D(i,k)-D'(i,i) &\geq \E\left[D(i,k)-D'(i,i)\right]-\sqrt{\frac{3\log n}{2}\cdot n\cdot 8^{2}}\\
&\geq \E\left[D(i,k)-D'(i,i)\right]-10\sqrt{n\log n}.\qedhere
\end{align*}
\end{proof}
For the analysis of the expectation of $D(i,k)-D'(i,i)$, we identify a class of ``most significant'' pairs among the good pairs.  
\begin{definition}\label{def:super_good}
A pair $(j,m)\in [n]\times\mathbb{Z}$ is said to be super-good if $I_{m}=\left[\frac{m}{L},\frac{m+1}{L}\right)$ satisfies 
$$\mu_{ij}-1\leq\frac{m}{L}<\frac{m+1}{L} \leq \mu_{ij}+1.$$
In particular, for any fixed $j\in [n]$ there are at least $L$ values of $m$ such that the pair $(j,m)$ is super-good. 
\end{definition}

The following lemma concerns only one or two Gaussian variables. Its proof is elementary and straightforward, so we provide only a brief sketch here.

\begin{lemma}\label{lem:CGaussian}
Assume $L\geq 10^{5}\log n$ and $\rho\geq 1-10^{-8}L^{-2}$. If $(j,m)$ is a good pair, then we have
\begin{enumerate}[label = (\arabic*)]
\item \label{lem:Gaussian1} $\PP[A_{ij}(m)=1] \leq 1/L$.
\item \label{lem:Gaussian2} $100\PP[A_{ij}(m)=1, B_{ij}(m)=0] \leq \PP[A_{ij}(m)=1]$ if $n$ is sufficiently large.
\item \label{lem:Gaussian3}
$\PP[A_{ij}(m)=1] \geq (10 L)^{-1}$ if $(j,m)$ is super-good.
\end{enumerate}
\end{lemma}
\begin{proof}[Proof Sketch]
Statements (1) and (3) are trivial by the Gaussian density function and \Cref{def:super_good}. We then sketch a proof of statement (2). Define a sub-interval  
$$J_{m}=\left[\frac{m}{L}+\frac{1}{800L},\frac{m+1}{L}-\frac{1}{800L}\right]\subset I_{m}.$$
Write the bivariate normal vector $(G_{ij},H_{ij})$ as $\left(X,\rho_{ij}X+\sigma_{ij}Y\right)$, where $X$ and $Y$ are independent Gaussian variables and $\sigma_{ij}=\sqrt{1-\rho_{ij}^{2}}$. So
\begin{equation}\label{eq:C_lemma9_break}
\frac{\PP[A_{ij}(m)=1,B_{ij}(m)=0]}{\PP[A_{ij}(m)=1]}\leq \frac{\PP[X \in I_{m}\setminus J_{m}]}{\PP[X\in I_{m}]}+\frac{\PP[X\in J_{m},\,\rho_{ij}X+\sigma_{ij}Y\not\in I_{m}]}{\PP[X\in J_{m}]}.
\end{equation}
Then it is straightforward to compute from the Gaussian density function that
\begin{equation}\label{eq:C_lemma9_1}
\frac{\PP[X \in I_{m}\setminus J_{m}]}{\PP[X\in I_{m}]}\leq \frac{1}{200},
\end{equation}
and that when $n$ is sufficiently large
\begin{equation}\label{eq:C_lemma9_2}
\PP\left[\rho_{ij}X+\sigma_{ij}Y\not\in I_{m}\middle| X\in J_{m}\right]\leq \PP\left[|\sigma_{ij}Y|\geq\frac{1}{10^{3} L}\right]\leq \frac{1}{200}.
\end{equation}
Combining \eqref{eq:C_lemma9_break}, \eqref{eq:C_lemma9_1} and \eqref{eq:C_lemma9_2} then yields statement (2).
\end{proof}
\begin{lemma}[Expectation]\label{lem:Cexpectation}
Assume $L\geq 10^{5}\log n$ and $\rho\geq 1-10^{-8}L^{-2}$. If $n$ is sufficiently large, we have 
$$
\E\left[D(i,k)-D'(i,i)\right] \geq 12\sqrt{n\log n}.$$
\end{lemma}
\begin{proof}
For each good pair $(j,m)$, set $p_{jm}=\PP[A_{ij}(m)=1]$ and $q_{jm}=\PP[A_{ij}(m)=1, B_{ij}(m)=0]$. By \Cref{lem:CGaussian}\ref{lem:Gaussian2} we have $100q_{jm}\leq p_{jm}\leq 1/2 $. For other pairs $(j,m)$, set $p_{jm}=q_{jm}=0$.

Note that for each $m\in\mathbb{Z}$, the $2n$ variables $A_{i1}(m),\dots,A_{in}(m),B_{k1}(m),\dots,B_{kn}(m)$ are mutually independent. By \eqref{eq:Cdef_of_D} and then \Cref{lem:Bern_to_sym} we have
\begin{equation}\label{eq:Cexpec_1}
\E[D(i,k)]\geq \sum_{m \in \mathbb{Z}} \inf_{x \in \mathbb{R}}\E\left|\sum_{j:(j,m)\text{ good pair}}A_{ij}(m)-x\right|\geq \frac{1}{2}\sum_{m\in\mathbb{Z}}f_{n}(p_{1m},\dots,p_{nm}),
\end{equation}
where we used that $\min\{p_{jm},1-p_{jm}\}=p_{jm}$ for every $(j,m)$. 

For each $m\in\mathbb{Z}$, the variables $A_{i1}(m)-B_{i1}(m),\dots, A_{in}(m)-B_{in}(m)$ are mutually independent symmetric $\{-1,0,1\}$ variables. Thus, applying \eqref{eq:Cdef_of_Dprime}, \Cref{lem:compare_f}, \Cref{lem:f_monotone} and \Cref{lem:CGaussian}\ref{lem:Gaussian2} in order, we obtain the following series of transitions (we recall that $f_n$ is defined in \Cref{def:Bfunction}):
\begin{equation}\label{eq:Cexpec_2}
\E[D'(i,i)]=\sum_{m\in\mathbb{Z}}f_{n}(q_{1m},\dots, q_{nm})\leq \frac{1}{4}\sum_{m\in\mathbb{Z}}f_{n}(16q_{1m},\dots,16q_{nm})\leq \frac{1}{4}\sum_{m\in\mathbb{Z}}f_{n}(p_{1m},\dots,p_{nm}).
\end{equation}
Due to \eqref{eq:Cexpec_1} and \eqref{eq:Cexpec_2}, it suffices to prove that 
$$\sum_{m\in\mathbb{Z}}f_{n}(p_{1m},\cdots,p_{nm}) \geq 50\sqrt{n \log n}.$$
For every $m\in\mathbb{Z}$ define $E_{m}:=\{j\in [n]: (j,m)\text{ is a super-good pair}\}$. 
Now we have
\begin{align*}
&\quad\sum_{m\in\mathbb{Z}}f_n(p_{1m},\cdots,p_{nm})\\
&\geq \sum_{m\in \mathbb{Z}}f_n\left(\frac{1}{10L}\mathbbm{1}\{1\in E_m\},\dots,\frac{1}{10L}\mathbbm{1}\{n\in E_m\}\right)&(\text{by \Cref{lem:f_monotone} and \Cref{lem:CGaussian}\ref{lem:Gaussian3}})\\
&\geq \sum_{m\in \mathbb{Z}}\frac{|E_{m}|}{5\sqrt{nL+L^{2}}}=\sum_{j=1}^{n}\sum_{m\in\mathbb{Z}}\frac{\mathbbm{1}\{j\in E_{m}\}}{5\sqrt{nL+L^{2}}}&(\text{by \Cref{lem:control_f}})\\
&\geq \sum_{j=1}^{n}\frac{L}{5\sqrt{nL+L^{2}}}=\frac{n}{5}\sqrt{\frac{L}{n+L}}\geq 50\sqrt{n\log n}&(\text{by \Cref{def:super_good}}),
\end{align*}
where we used $L\geq 10^{5}\log n$ and that $n$ is sufficiently large in the last estimate.
\end{proof}
The proof of \eqref{eq:matrixmatching} now follows easily.
\begin{proof}[Proof of \eqref{eq:matrixmatching}]
Combining \Cref{lem:Cconcentration,lem:Cexpectation} we have
$$\PP\left[D(i,k)-D'(i,i)\leq 2\sqrt{n\log n}\right] \leq O(n^{-3}).$$
From \Cref{lem:Ctrunc} we then have
\[\PP\left[D(i,k)-D(i,i)\leq 2\sqrt{n\log n}\right]\leq O(n^{-3})+\exp\left(-\Omega\left(\log^{4/3}n\right)\right)=O(n^{-3}).\qedhere\]
\end{proof}

\section{Numerical Experiments}\label{sec:numerical}
In this appendix, we conduct an empirical evaluation of our algorithm. The codes of the experiments are available at \href{https://github.com/yuanzheng-wang/code/tree/main/inhomogeneous}{\textcolor{blue}{https://github.com/yuanzheng-wang/code/tree/main/inhomogeneous}}. While our theoretical analysis used a specific distance function as in \eqref{eq:defofdistance}, we include a broader class of distance functions in this empirical study.

We first consider the 1-Wasserstein distance, which is the same distance used in the empirical experiments of \cite[Section 5]{DMWX21}. Recall that $X_{j}$ (respectively $Y_{j}$) denotes the degree of the vertex $j$ in $G$ (respectively $H$). Additionally, we introduce $N_{G}(i)$ and $N_{H}(i)$ as the sets of neighbors of the vertex $i$ in $G$ and $H$, respectively. We define the distance
$$D^{\mathsf{ref}}(i,k):=\int_{\mathbb{R}}\left|\frac{1}{|N_{G}(i)|}\sum_{j\in N_{G}(i)}\mathbbm{1}\{X_{j}\leq t\}-\frac{1}{|N_{H}(k)|}\sum_{j\in N_{H}(k)}\mathbbm{1}\{Y_{j}\leq t\}\right|\d t.$$
This is the $L^{1}$-distance between the empirical cumulative distribution functions of the sample sets $\{X_{j}\}_{j\in N_{G}(i)}$ and $\{Y_{j}\}_{j\in N_{H}(k)}$, and is the distance employed in the numerical experiments of \cite{DMWX21}. We also define the following analogous distance function where we take square roots of degrees:
$$D^{\mathsf{cdf}}(i,k):=\int_{\mathbb{R}}\left|\frac{1}{|N_{G}(i)|}\sum_{j\in N_{G}(i)}\mathbbm{1}\left\{\sqrt{X_{j}}\leq t\right\}-\frac{1}{|N_{H}(k)|}\sum_{j\in N_{H}(k)}\mathbbm{1}\left\{\sqrt{Y_{j}}\leq t\right\}\right|\d t.$$

We next consider balls-into-bins-type distances featured in \Cref{alg:graphmatching}. Recall from \Cref{subsec:intro_comparison} that we have used ``cyclic intervals'' as bins to facilitate the theoretical analysis. In the numerical experiments we simply use the usual intervals as bins. Fix a real number $r>0$ and let $I_{t}=[t-r,t+r]$ for all $t\in\mathbb{R}$. Define
$$D^{\mathsf{bin}(r)}(i,k):=\int_{\mathbb{R}}\left|\sum_{j\in N_{G}(i)}\mathbbm{1}\left\{\sqrt{X_{j}}\in I_{t}\right\}-\sum_{j\in N_{H}(k)}\mathbbm{1}\left\{\sqrt{Y_{j}}\in I_{t}\right\}\right|\d t.$$

\subsection{Numerical study on simulated data}\label{subsec:graph models experiments}
In this subsection, we conduct simulations on three examples of \Cref{def:graphensemble}, where all graphs have 1000 vertices. We will compare the performances of the above variations of the Degree-Prorile algorithm.

\Cref{fig:ER-graph} demonstrates the performance of \Cref{alg:graphmatching} (under different distance functions) in the correlated Erd\H{o}s-R\'{e}nyi model $\mathbb{G}(n,p;1-\delta)$, which is a special case of the model in \Cref{def:graphensemble} with $p_{ij}=p$ and $\delta_{ij}=\delta$ for all $i,j\in [n]$ such that $i\neq j$.
\begin{figure}
\begin{center}
\includegraphics[scale=0.95]{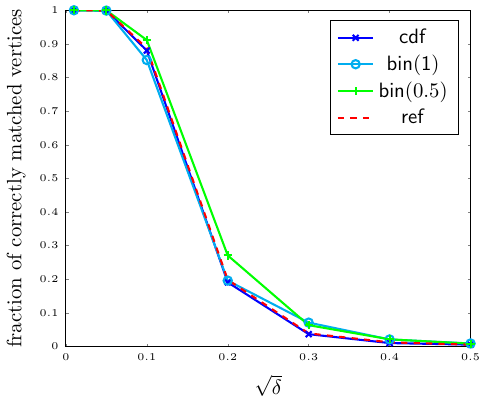}
    \caption{A simulation on correlated Erd\H{o}s-R\'{e}nyi model $\mathbb{G}(n,p;1-\delta)$ with $n=1000$, $p=0.05$ and varying $\sqrt{\delta}$. For each value of $\sqrt{\delta}$, the fraction of correctly matched vertices shown is averaged over 10 independent runs.}
\label{fig:ER-graph}
\end{center}
\end{figure}
Note that in \Cref{fig:ER-graph}, the line produced by the distance function $D^{\mathsf{ref}}$ almost coincides with that produced by $D^{\mathsf{cdf}}$. This corresponds to the fact that in Erd\H{o}s-R\'{e}nyi graphs the degree of each vertex has the same variance, in which case the ``partial standardization'' trick (i.e. taking square roots, as discussed in \Cref{subsec:intro_comparison}) seems unnecessary. 

We then proceed to two inhomogeneous random graph models, which are special cases of \Cref{def:graphensemble}. In the notation of \Cref{def:graphensemble}, the first model has parameters

\begin{equation}\label{eq:model1}
p_{ij}=\frac{1}{2}|i-j|^{-1/2}\text{ and } \delta_{ij}=\delta,\quad\forall i,j \in [n] \text{ such that } i\neq j.
\end{equation}
for varying $\delta$. In \Cref{fig:model1}, we demonstrate the performance of \Cref{alg:graphmatching} (under the same four distance functions) for the above model. The performance of the distance function $D^{\mathsf{ref}}$ remains nearly identical to that of $D^{\mathsf{cdf}}$, but is slightly worse than those of the balls-into-bins-type distances with suitably chosen bin sizes.

\begin{figure}
\begin{center}
\includegraphics[scale=0.95]{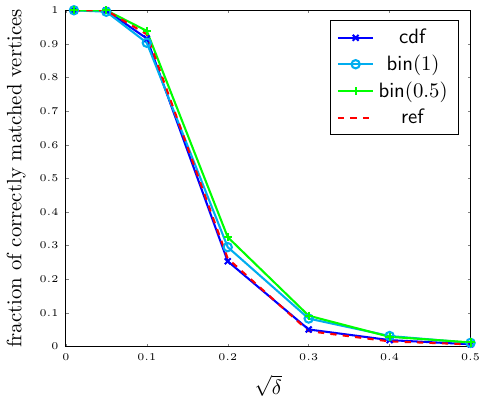}
    \caption{A simulation on correlated inhomogeneous random graph pairs $(G,H)$ following the distribution of \Cref{def:graphensemble} with parameters given in \eqref{eq:model1}. For each value of $\sqrt{\delta}$, the fraction of correctly matched vertices shown is averaged over 10 independent runs.}
\label{fig:model1}
\end{center}
\end{figure}

Finally, we consider the Chung-Lu model introduced in \cite{CL02}, a canonical random graph model where we may specify a sequence of expected degrees. To simulate a scale-free network, we specify the degree sequence by independently sampling its elements from a power law distribution.

We now give the precise definition of our model. Let $\gamma \in (2,3)$ and $b>0$ be two fixed parameters. Let $w_1,\dots,w_n$ be independent random variables, each of which follows a power law distribution in $[b,\infty)$ with density proportional to $x^{-\gamma}$. Let $\delta$ be a positive constant in $(0,1)$. Conditioned on the $n$ random variables $w_1,\dots,w_n$, the pair $(G,H)$ follows the distribution in \Cref{def:graphensemble} with parameters
$$
p_{ij}=\min\left\{\frac{w_{i}w_{j}}{\sum_{k=1}^{n}w_{k}},1\right\}\text{ and }\delta_{ij}=\delta,\quad \forall i,j \in [n] \text{ such that } i\neq j.
$$
In our simulations, $\delta$ is varying, and the parameters $\gamma$ and $b$ are chosen as
\begin{equation}\label{eq:model2}
\gamma=2.5 \text{ and } b=10^{4/3},
\end{equation}
which ensures an average degree of approximately 50.

The results of running \Cref{alg:graphmatching} on this model are presented in \Cref{fig:model2}. In comparison to the previous two models, the difference in performance between the distances is noticeably larger, potentially due to the greater variation in vertex degrees within the current model. Notably, the distances based on balls-into-bins-type signatures now significantly outperforms the CDF-based distances. To provide a more comprehensive comparison, we also include the distance based on discrete bins, as utilized in the theoretical analysis of \cite{DMWX21}. This distance is defined as follows. For fixed real number $r>0$, define
$$D^{\mathsf{disc}(r)}(i,k):=\sum_{m\in\mathbb{Z}}\left|\sum_{j\in N_{G}(i)}\mathbbm{1}\left\{2mr\leq \sqrt{X_{j}}<2(m+1)r\right\}-\sum_{j\in N_{H}(k)}\mathbbm{1}\left\{2mr\leq \sqrt{Y_{j}}<2(m+1)r\right\}\right|.$$
This means we choose the bins to be a partition of the real line into intervals of length $2r$. In the case $r=0.5$, we observe from \Cref{fig:model2} that the algorithm with continuous bins (i.e., using $D^{\mathsf{bin}(r)}$) performs significantly better than this discrete analogue (i.e., using $D^{\mathsf{disc}(r)}$).

\begin{figure}[H]
\begin{center}
\includegraphics[scale=0.9]{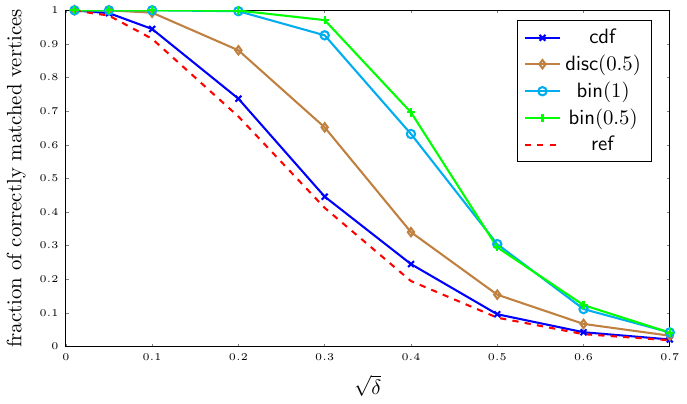}
    \caption{A simulation on the power-law-degree-adapted Chung-Lu model with parameters given in \eqref{eq:model2}. For each value of $\sqrt{\delta}$, the fraction of correctly matched vertices shown is averaged over 10 independent runs.}
\label{fig:model2}
\end{center}
\end{figure}

\subsection{Numerical study on real data}
\label{subsec:real graphs} 
In this subsection, we conduct numerical study on real data, and evaluate the performances of the degree-profile-based algorithms as described in \Cref{subsec:graph models experiments}. Another quadratic programming algorithm based on convex relaxation (QP) \cite{ABK15,FS15} is also experimented for comparison. In QP, we first minimize the Frobenius norm $||GX-XH||_{F}$, where $G,H$ are adjacency matrices of graphs and $X$ ranges over the class of doubly-stochastic matrices (i.e., the convex hull of all the permutation matrices). After solving this quadratic programming problem, we round the solution $\hat{X}$ to a permutation matrix by maximizing $||\Pi-\hat{X}||_{F}$ for $\Pi\in\mathfrak{S}_{n}$, which is essentially a linear assignment problem and can be solved via the Hungarian method.

It turns out that QP is much more computationally expensive then the degree-profile-based algorithms. In our experiment for QP, we adopt the fast solver developed in \cite{DMWX21}, which was based on alternating direction method of multipliers (ADMM) algorithm in \cite{BCP+11}. Still, we will see that QP is much more time-consuming. 

Inspired by \cite{KHG15, DMWX21}, we first consider the Slashdot network, which consists of connections among users of Slashdot, a technology-related news website. Obtained in February 2009, this network is accessible through the Stanford Large Network Database \cite{SNA09}.

As in \cite{DMWX21}, we obtain the pair of correlated graphs via subsampling. The parent graph $A$ is constructed by restricting to the subnetwork induced by users with ID smaller than 750, and connect users $i$ and $j$ if and only if there is a directed link between them. Then $A$ is a deterministic graph with 750 vertices and 3419 edges.\footnote{The edge statistics here is slightly different from that of 3338 edges in \cite{DMWX21}, but we believe this is immaterial.} The two correlated graphs $G$ and $H$ to be matched are then obtained by independently subsampling the edges of $A$ with probability $s$, with vertices of $H$ being permuted uniformly at random. The edge subsampling probability $s$ ranges from $0.6$ to $1$, and we do 10 independent runs for each $s$.

\begin{figure}[h]
\begin{center}
\includegraphics[scale = 0.5]{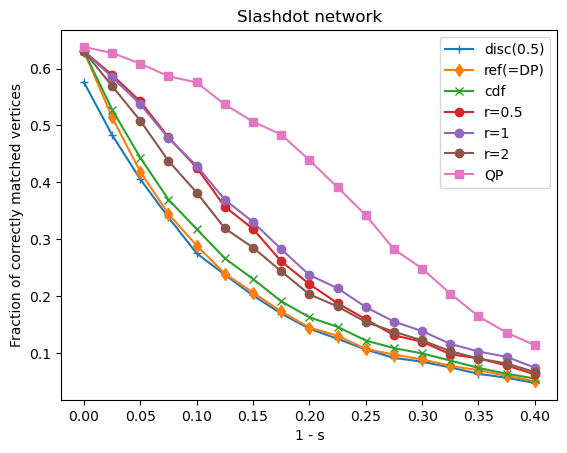}
    \caption{A simulation over pairs of subsampled graphs from the Slashdot network. For each value of $s$, the fraction of correctly matched vertices shown is averaged over 10 independent runs.}
\label{fig:Slashdot Network}
\end{center}
\end{figure}

The results for the experiments are presented in \Cref{fig:Slashdot Network}. Note that even in the noiseless case $s = 1$, the accuracy rates of all the algorithms are only no more than around 0.63; this is due to the fact that our parent graph $A$ has many groups of low-degree (say, 1 or 2) vertices attached to the same high-degree vertices, creating symmetries.

With the graph $A$ being largely inhomogeneous, we can see that in the noisy regime, our algorithm (the three dotted lines) is still the best among all the variations of degree-profile algorithms (including ref, which we recall is the distance employed in the numerical experiments of \cite{DMWX21}), aligning with our theoretical analysis. It is obvious that QP outperforms all the degree-profile algorithms in terms of recovery fractions. However, the degree-profile algorithms have the advantage on the running time: to complete the independent 10 runs, disc, ref(=DP) and cdf take about 10 minutes, our algorithm takes around 21 minutes, while QP costs five and a half hours.

We next consider the networks of Autonomous Systems, drawing inspirations from \cite{FMWX22a}. The dataset is from the University of Oregon Route Views Project \cite{Oregon}, and is available at the Stanford Large Network Database \cite{SNA09,LKF05}.

The data contains networks of Autonomous Systems observed on nine days between March 31, 2001 and May 26, 2001, one week apart for each two neighbouring days. Edges and (a small fraction of) vertices were added and deleted over time; the number of vertices varied between 10,670 and 11,174, while the number of vertices varied between 22,002 and 23,409. Thus all the nine networks can be seen as noisy versions of the first network. We then consider the subgraphs of the networks induced by 10,000 vertices that are present on all the nine days with the largest degrees. We apply independent uniform random permutations on them, and then match them to the first subgraph via the algorithms.

We still expect QP to be the best among all the algorithms in terms of recovery fractions. However, we did not complete running QP because it is too computationally costly (even with the faster solver version) on large graphs. On a Dell Inspiron 7590 laptop, 24 hours is already not enough for QP to match any two of the subgraphs above. For comparison, it takes about one hour to run disc, ref(=DP), and cdf, and two hours to run our algorithm.

The results for all algorithms except QP are displayed in \Cref{subfig:oregon whole}. Although our algorithm (the three dotted lines) still outperforms other variations of DP, all the correction rates dropped significantly in the noisy regime. In addition, even in the noiseless situation, that is, matching the first subgraph to itself, the fraction of correctly matched vertices are yet less than 0.4. These results are due to the same reason as in the Slashdot Network, i.e.~large non-trivial symmetries generated by the numerous small-degree vertices connected to the same high-degree vertices.

We thus turn to the high-degree vertices of the networks of Autonomous Systems as in \cite{FMWX22a}, and consider the top 1,000 vertices with the highest degrees in the first subgraph. With the same matchings on the subgraphs of 10,000 vertices in hand, we evaluate the effectiveness of algorithms by measuring their correctness only on those high-degree vertices. The results are shown in \Cref{subfig:oregon high-degree}. Our algorithm produces substantially better results than others.

\begin{figure}
\centering
\begin{subfigure}{.45\textwidth}
  \centering
  \includegraphics[scale=0.4]{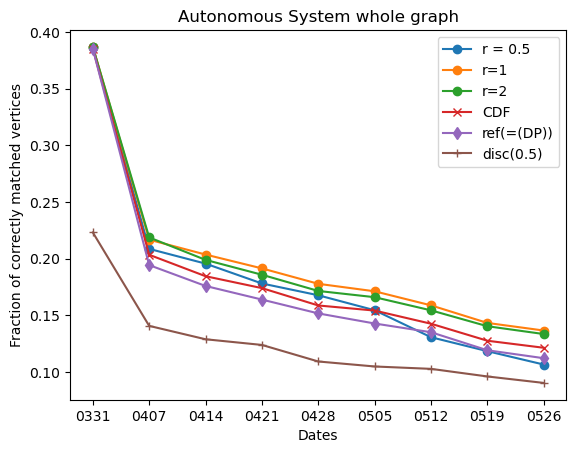}
  \caption{Fraction of correctly matched vertices in whole graph of 10,000 vertices.}
  \label{subfig:oregon whole}
\end{subfigure}%
\quad
\begin{subfigure}{.45\textwidth}
  \centering
  \includegraphics[scale=0.4]{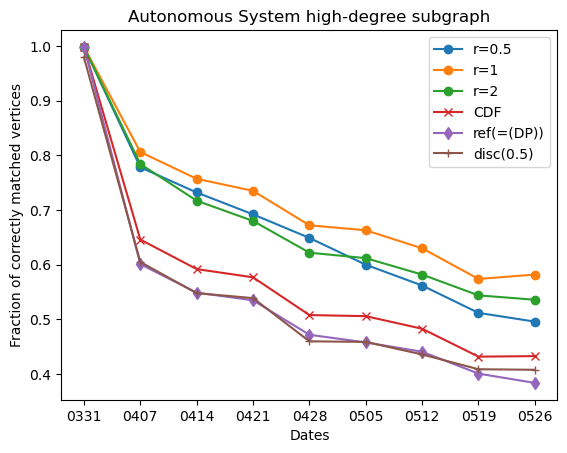}
  \caption{Fraction of correctly matched vertices in high-degree subgraph of 1,000 vertices.}
  \label{subfig:oregon high-degree}
\end{subfigure}
\vspace{0.5cm}
\caption{A simulation over matching networks of Autonomous Systems on nine days to that on the first day.}
\vspace{0.5cm}
\label{fig:oregon}
\end{figure}
\end{document}